\documentclass[12pt,a4paper,reqno,french,tikz]{amsart}
\usepackage[utf8]{inputenc} \usepackage[T1]{fontenc}
\usepackage{amsthm,amsmath,amsfonts,amssymb,amsxtra,appendix,bookmark,dsfont,bm,mathrsfs,amstext,amsopn,mathrsfs,mathtools,comment,cite,hyperref,color,xcolor,cite,graphicx}
\usepackage{fullpage}
\usepackage{changepage}
\usepackage{url}

\usepackage{tkz-euclide}
\usepackage{outlines}
\usepackage{makecell}
\usepackage{tikz,tikz-cd,pgfplots,pgfplotstable}
\usetikzlibrary{datavisualization}
\usetikzlibrary{datavisualization.formats.functions}
\usetikzlibrary{shapes.callouts}
\tikzset{
  level/.style   = { ultra thick, blue },
  connect/.style = { dashed, red },
  notice/.style  = { draw, rectangle callout, callout relative pointer={#1} },
  label/.style   = { text width=2cm }
}

\newtheorem{theorem}{Theorem}[section]

\newtheorem{lemma}[theorem]{Lemma}

\newtheorem{proposition}[theorem]{Proposition}
\newtheorem{corollary}[theorem]{Corollary}

\newtheorem{remark}[theorem]{Remark}
\newtheorem{assumption}[theorem]{Assumption}

\let\C\relax\newcommand{\C}{\mathbb{C}}\newcommand{\Z}{\mathbb{Z}}\newcommand{\R}{\mathbb{R}}\newcommand{\N}{\mathbb{N}}\newcommand{\D}{\mathbb{D}}
\newcommand{\bG}{\text{\boldmath$G$}}
\newcommand{\bx}{\text{\boldmath$x$}}
\newcommand{\bU}{\text{\boldmath$U$}}

\newcommand{\bbP}{\mathbb{P}}
\newcommand{\bbL}{\mathbb{L}_z}
\newcommand{\bbH}{\mathbb{H}_z}
\newcommand{\bbA}{\mathbb{A}_z}

\newcommand{\bbzL}{\mathbb{L}_{z,\bU}}
\newcommand{\bbzH}{\mathbb{H}_{z,\bU}}
\newcommand{\bbzA}{\mathbb{A}_{z,\bU}}
\newcommand{\bal}{\text{\boldmath$\alpha$}}

\newcommand\cB{\mathcal{B}}\newcommand\cC{\mathcal{C}}\newcommand\cD{\mathcal{D}}\newcommand\cG{\mathcal{G}}\newcommand\cH{\mathcal{H}}\newcommand\cI{\mathcal{I}}\newcommand\cL{\mathcal{L}}\newcommand\cM{\mathcal{M}}\newcommand\cN{\mathcal{N}}\newcommand\cT{\mathcal{T}}\newcommand\cU{\mathcal{U}}

\DeclareMathOperator{\tr}{Tr}\DeclareMathOperator{\Ker}{Ker}
\def\d{{\rm d}}\def\1{{\mathds{1}}}

\renewcommand{\ge}{\geqslant}\renewcommand{\le}{\leqslant}
\newcommand{\floor}[1]{\lfloor #1\rfloor}

\newcommand{\pa}[1]{\left( #1 \right)} 
\newcommand{\acs}[1]{\left\{ #1 \right\}} 
\newcommand{\seg}[1]{\left[ #1 \right]} 
\newcommand{\ab}[1]{\left|#1\right|} 
\newcommand{\ps}[1]{\left< #1 \right>} 
\newcommand{\nor}[2]{ \left| \! \left| #1 \right| \! \right|_{#2} } 
\newcommand{\norm}[1]{ \left| \! \left| #1 \right| \! \right| } 

\newcommand\vp{\varphi} 
\newcommand{\ep}{\varepsilon} 
\let\p\relax\newcommand{\p}{\psi} 

\newcommand{\f}[2]{\frac{#1}{#2}} 
\newcommand{\mymax}[1]{\underset{\substack{#1}}{\text{\normalfont{max}}}\;} 
\newcommand{\mymin}[1]{\underset{\substack{#1}}{\text{\normalfont{min}}}\quad} 
\newcommand{\argmin}[1]{\underset{\substack{#1}}{\text{\normalfont{argmin}}}\quad} 
\newcommand{\mysup}[1]{\underset{\substack{#1}}{\text{\normalfont{sup}}}\quad} 
\newcommand{\ind}[1]{_{\textup{#1}}} 
\newcommand{\hzmp}{\pa{H^0 + \mu}^{\f \kappa 2}}
\newcommand{\hzmm}{\pa{H^0 + \mu}^{-\f \kappa 2}}

\newcommand{\thh}{^{\text{th}}} 
\newcommand{\Om}{\Omega_{a}^{\lambda\ind{min} , \lambda\ind{max}}}
\DeclareMathOperator{\Gen}{Gen}

\DeclareMathOperator{\Aff}{Aff}

\newcommand{\bs}[1]{\textcolor{blue}{#1}}

\newcommand{\Del}{(1-s_\bal)}
\newcommand{\Delm}{(s_\bal -1)}
\newcommand{\co}{\mathfrak{C}}
\newcommand{\nv}[1]{#1}
\newcommand{\sol}{P}
\newcommand{\pt}{\cT}
\newcommand{\gw}{W}

\author{Geneviève Dusson}
\address[Geneviève Dusson]{Laboratoire de Mathématiques de Besançon, UMR CNRS 6623,
Université de Franche-Comté, 16 route de Gray, 25030 Besançon, France} 
\email{genevieve.dusson@math.cnrs.fr}

\author{Louis Garrigue}
\address[Louis Garrigue]{Laboratoire ``analyse géométrie modélisation'', CY Cergy Paris Université, 95302 Cergy-Pontoise, France} 
\email{louis.garrigue@cyu.fr}

\author{Benjamin Stamm}
\address[Benjamin Stamm]{Institute of Applied Analysis and Numerical Simulation, University of Stuttgart, 70569 Stuttgart, Germany
}
\email{benjamin.stamm@mathematik.uni-stuttgart.de}

\title[Multipoint perturbation]{A multipoint perturbation formula\\for eigenvalue problems}
\date{\today}
\begin{document}

\maketitle

\begin{abstract}
Standard perturbation theory of eigenvalue problems consists of obtaining approximations of eigenmodes in the neighborhood of an operator where the corresponding eigenmode is known. Nevertheless, if the corresponding eigenmodes of several nearby operators are known, standard perturbation theory cannot simultaneously use all this knowledge to provide a better approximation. We derive a resolvent formula enabling such an approximation result, and provide numerical examples for which this method is more competitive than standard perturbation theory.
\end{abstract}

\section{Introduction}

Eigenvalue problems are ubiquitous in applied mathematics, for instance in quantum physics, structural mechanics, graph theory and optimization. 
Numerical methods to solve these eigenvalue problems are hence key to obtain accurate approximate solutions in an efficient and robust way. Let us take a Hilbert space, consider $H^0$ to be a self-adjoint operator on this Hilbert space, and define a family of self-adjoint operators
\begin{align*}
    H(G) := H^0 + G
\end{align*}
called Hamiltonians, where $G$ is a self-adjoint operator. The operators $H(G)$ should be understood in the quadratic form sense, and we assume them to be bounded from below. We consider the eigenvalue problem associated to $H(G)$ under variation of $G$. A widely used approach to treat eigenvalue problems is based on perturbation theory,
where one knows or can easily compute the eigenstates of an operator $H(G_1)$ for some $G_1$ and deduces properties of approximate solutions for the close problem $H(G_1) + \gw$, where $\gw$ is small in some sense to be made precise later, see~\cite{Kato,Gaeta23,Murdock99}, \cite[Section XII]{ReeSim4} and~\cite[Chapters 3 and 5]{Lewin24} for a presentation of perturbation theory for eigenvalue problems.  
In other words, from spectral information on $H(G_1)$, one can deduce information on the perturbed $H(G_1) + \gw$ for small perturbations $\gw$, thus locally around $H(G_1)$.

In contrast, when the solution of an eigenvalue problem with Hamiltonian $H(G_j)$ is known at $n$ points $\pa{G_j}_{j=1}^n$, standard perturbation theory does not exploit the full approximation power of using the information at all $n$ points to obtain approximations at a new point $G$ as it will only use local information of the closest point $G_j$. 
We illustrate the comparison between single point and multipoint perturbation theory on Figure~\ref{fig:illustr2_pt_and_multipoint}.

\begin{figure}[h]
\begin{tikzpicture}[scale=1]
\tkzDefPoint(-3,2){G1}
\tkzDefPoint(-1.5,2.2){G2}
\tkzDefPoint(1,1){G3}
\tkzDefPoint(2.5,-1){G4}
\tkzDefPoint(-2,-1.5){G5}
\tkzDefPoint(0,0){G}
\draw [add= -0.0 and -0.1,red,-latex]  (G1) to (G);
\draw [add= -0.0 and -0.1,red,-latex]  (G2) to (G);
\draw [add= -0.0 and -0.1,red,-latex]  (G3) to (G);
\draw [add= -0.0 and -0.1,red,-latex]  (G4) to (G);
\draw [add= -0.0 and -0.1,red,-latex]  (G5) to (G);
\draw [->,blue] (G3.south) to [out=-60,in=0] ($(G.south)+(0.2,-0.)$);
\foreach \n in {G1,G2,G3,G4,G5}
  \node at (\n)[circle,fill,inner sep=1.5pt]{};
  \foreach \n in {G}
  \node at (\n)[circle,fill,brown,inner sep=2pt]{};
\tkzLabelPoint[left](G1){$G_1$}
\tkzLabelPoint[right,left](G2){$G_2$}
\tkzLabelPoint[right](G3){$G_3$}
\tkzLabelPoint[right](G4){$G_4$}
\tkzLabelPoint[right,left](G5){$G_5$}
\tkzLabelPoint[left](G){$G$}
\node[color=blue] at (3.8,0.2) {Standard perturbation theory};
\node[color=red] at (-3,-0.2) {Multipoint perturbation};
\end{tikzpicture}
\caption{In standard perturbation theory, one gets approximations of the solution $\sol(G)$ at a new point $G$ by using only the closest $G_j$ on which we know the solution $\sol(G_j)$, while in multipoint perturbation we simultaneously use all the $G_j$'s on which we know the solutions.}\label{fig:illustr2_pt_and_multipoint}
\end{figure}

Other approaches for estimating spectral information of a given operator consist in using commutator relations with different operators such as in~\cite{LevPar02}, where upper bounds on the eigenvalue gap of abstract self-adjoint operators are obtained.

Let us now detail our framework. We consider some admissible set $\cG$ of self-adjoint operators $G$'s, which will be $H^0$-form bounded as will be precised later in~\eqref{eq:set_operators}. For $G \in \cG$ and $k \in \N\backslash\{0\}$, we denote by $\sol\pa{G}$ the spectral projection onto the eigenspace corresponding to the $k\thh$ lowest eigenvalue of the operator $H(G)$, which we assume to be non-degenerate. It is clear that the knowledge of $\sol(G)$ is equivalent to know an eigenvector corresponding to the $k\thh$ eigenvalue. 
Our framework can be generalized to projectors onto the space generated by several eigenvalues, and those eigenvalues could lie anywhere in the discrete spectrum since our methodology relies on spectral projectors, or, in other words, on density matrices for the selected states. 
However, this goes beyond the scope of this work and we restrict our analysis within this article to the non-degenerate case. 

The motivation for this work comes from~\cite{PolMikDus20,PolDusSta21}, where it is observed that if $G$ is close to a linear combination of $G_j$'s, that is
\[
	G \simeq \sum_{j=1}^n \alpha_j G_j,
\]
then the density matrix appears to be well approximated by the same linear combination of the corresponding density matrix, i.e.
\[
	\sol(G) \simeq \sum_{j=1}^n \alpha_j \sol(G_j).
\]
This amounts to saying that the map $G \mapsto \sol(G)$ is locally very close to linear, under some conditions.
In this work, we aim at partly explaining this phenomenon developing a multipoint perturbation theory (a precise meaning will be given later). 

%

This article is organized as follows. In Section~\ref{sec:Definitions}, we provide the mathematical setting and present a review of standard perturbation theory. In Section~\ref{sec:main_sec}, we present our main result, the resolvent formula~\eqref{eq:main_formula}, and compute the first orders in the expansion of $\sol(G)$ in Section~\ref{sec:Firsts order}. We illustrate our work with numerical simulations on Schrödinger operators in Section~\ref{sec:simus}. We conclude that the produced approximations are particularly relevant when the $G_j$'s are close to each other, and when the $G$ on which one wants the solution is close to the affine space spanned by the $G_j$'s. All the proofs are provided in Section~\ref{sec:Proofs}.

\section{Mathematical setting}
\label{sec:Definitions}

In this section, we define the main mathematical objects that will be needed later on for the presentation of the multipoint perturbation method itself and the theoretical analysis that follows.

\subsection{Hamiltonians}%
\label{sub:Hamiltonian}

First, we consider a separable Hilbert space $\cH$ endowed with a scalar product $\ps{\cdot,\cdot}$, with corresponding norm $\|\cdot\|$. Inequalities on operators will be considered in the sense of forms, meaning that for two operators $A$ and $B$, we write $A \le B$ if there is a vector subspace $\cH_0 \subset \cH$, dense in $\cH$, such that $\ps{\psi, A \psi} \le \ps{\psi, B \psi}$ for all $\psi \in \cH_0$. We consider a self-adjoint operator $H^0$ on $\cH$, satisfying $H^0 \ge a$ for some $a \in \R_+$.
Throughout the paper, $H^0$ will be a fixed operator, for instance $-\Delta$ in the case of Schrödinger operators. Moreover, for a given operator $A$, we denote its adjoint operator by $A^*$.

We define the \emph{set of admissible \nv{operators{}}} by
\begin{align}\label{eq:set_operators}
\cG := \big\{G \text{ linear operator on } \cH, 
\; G^* = G, \; \forall \ep > 0 \; \exists c^G_\ep > 0, \; \ab{G} \le \ep H^0  + c^G_\ep \big\},
\end{align}
where $|A| := \sqrt{A^* A}$ is defined via the spectral theorem. \nv{The inequality in~\eqref{eq:set_operators} has to be understood in the sense of forms, meaning that for any $\p$ in the form domain of $H^0$, $\ps{\p, \ab{G} \p } \le \ep \ps{\p, H^0 \p} + c^G_\ep \nor{\p}{}^2 $. 
Note that $G \in \cG$ is not necessarily a ``potential'', i.e. a multiplication operator by a function, and can also be a differential operator.{}}
For any $G \in \cG$, by the Kato--Lions--Lax--Milgram--Nelson (KLMN) theorem~\cite[p. 323]{ReeSim2}, the Friedrichs extension of
\begin{align*}
H(G) = H^0 + G
\end{align*}
is a well-defined self-adjoint and bounded from below operator. Considering the set~\eqref{eq:set_operators} enables one to take into account Schrödinger operators perturbed by, e.g., singular potentials. We could also treat operators $H(G)$ which are not bounded from below, this extension is immediate from our presentation. 

The finite-dimensional case is covered by our analysis as we make no assumption whether the Hilbert space is finite or infinite-dimensional. In the finite-dimensional case, $\cG$ is the space of $m \times m$ matrices, where $m$ is the dimension of the considered Hilbert space.

\subsection{Non-degeneracy assumption}\label{sub:non_deg}
We use the notation $\N = \{0,1,\dots\}$, so $0 \in \N$, and we will use the notation $\N \backslash \{0\}$ for the set of strictly positive integers.  Since we only consider non-degenerate eigenmodes, we take
$
k \in \N\backslash\{0\},
$
the targeted index of the eigenvalue, \textit{fixed} throughout all the document. Moreover, in all this article, for any $G \in \cG$ and any generic $p \in \N\backslash\{0\}$, we denote by $\lambda^p(G)$ the $p\thh$ eigenvalue of $H(G)$, counting multiplicities. \nv{To avoid unnecessarily overloading the notation, we do not explicitly write the $k$-dependence (labeling the different eigenmodes) for the density matrices, eigenvectors and eigenvalues, except when explicitely precised.{}}

For any $\pa{G_j}_{j=1}^n \in \cG^n$, we assume that the $k\thh$ eigenvalue is uniformly separated from the rest of the spectrum, and that it remains so as we navigate on the corresponding admissible space. To translate this into mathematical terms, we first define for any $\lambda\ind{min}, \lambda\ind{max} \in \R$, $\lambda\ind{min} < \lambda\ind{max}$, $a >0$, the set
	\nv{
\begin{multline}\label{hypo:non_deg_set}
\Om := \Big\{G \in \cG \; \Big| \; \sigma(H(G)) \cap  [\lambda\ind{min},\lambda\ind{max}] = \{\lambda^k(G)\}, \\
	\text{dist} \pa{\sigma(H(G)) \backslash \{\lambda^k(G)\} , [\lambda\ind{min},\lambda\ind{max}]} \ge a\Big\},
\end{multline}
{}}
where $\sigma(A)$ denotes the spectrum of $A$ for any operator $A$. We then formalize the following assumption.

\begin{assumption}
\label{as:convgi}
There holds
\begin{align}\label{hypo:non_deg}
	\nv{\exists \lambda\ind{min}, \lambda\ind{max} \in \R, \; \exists a > 0, \; \lambda\ind{min} < \lambda\ind{max}, \quad \mbox{such that} \quad \text{Conv } \pa{G_j}_{j=1}^n \subset \Om,{}}
\end{align}
where ``Conv'' denotes the convex hull.
\end{assumption}
Assumption~\ref{as:convgi} is illustrated in Figure~\ref{fig:assumption_nondeg_gap}. In the following, we will consider $G$'s such that
\begin{align*}
G \in \Om.
\end{align*}
Moreover, in all this document, $\cC \subset \C$ will denote a contour, say a rectangle, such that $\cC \cap \R = \{\lambda\ind{min}, \lambda\ind{max}\}$. From Assumption~\ref{as:convgi}, this contour will enclose exactly one eigenvalue of $H(G)$, namely the $k\thh$ eigenvalue, for any $G \in \Om$.

\begin{figure}[h!]
\begin{center}
\begin{tikzpicture}[scale=0.8]
  \draw[scale=1, domain=-3.5:3.5, smooth, variable=\x, blue] plot ({\x}, {0.1*\x - 0.07*\x*\x - 0.01*\x*\x*\x + 6.5});
  \draw[scale=1, domain=-3.5:3.5, smooth, variable=\x, blue] plot ({\x}, {-0.05*\x + 0.01*\x*\x - 0.001*\x*\x*\x + 6});
  \draw[scale=1, domain=-3.5:3.5, smooth, variable=\x, blue] plot ({\x}, {-0.3*\x - 0.1*\x*\x + 0.05*\x*\x*\x + 4.5});
  \draw[scale=1, domain=-3.5:3.5, smooth, variable=\x, blue] plot ({\x}, {0.2*\x + 0.03*(\x-1)*(\x-1) + 0.01*\x*\x*\x -0.2});
  \draw[scale=1, domain=-3.5:3.5, smooth, variable=\x, blue] plot ({\x}, {0.1*(\x-1) - 0.01*(\x-1)*(\x-1) - 0.01*\x*(\x-1)*(\x-1) + 1});
  \draw[->] (-3.4, 0) -- (3.4, 0) node[right] {$G$};
  \draw[-] (3, 0) -- (3, -0.1); \node at (3,-0.4) {$G_2$};
  \draw[-] (-3, 0) -- (-3, -0.1); \node at (-3,-0.4) {$G_1$};
  \draw[->] (0, -0.5) -- (0, 6.5); 
  \draw[-] (-3.5, 5.5) -- (3.5, 5.5) node[right] {$\lambda_{\text{max}} + a$};
  \draw[-] (-3.5, 5) -- (3.5, 5) node[right] {$\lambda_{\text{max}}$};
  \draw[-] (-3.5, 2) -- (3.5, 2) node[right] {$\lambda_{\text{min}}$};
  \draw[-] (-3.5, 1.5) -- (3.5, 1.5) node[right] {$\lambda_{\text{min}} - a$};
   \node[blue] at (1.5,3.5) {$\lambda^k(G)$};
   \node[blue] at (0,7) {$\sigma(H(G))$};
\end{tikzpicture}
\end{center}\label{fig:assumption_nondeg_gap}
\caption{We illustrate Assumption~\ref{as:convgi}, where we consider that uniformly in $G$, the $k^{\text{th}}$ level is non-degenerate and that it is strictly separated from the rest of the spectrum by a distance $a > 0$.}
\end{figure}

\subsection{Operator norms}%
\label{sub:Norms}

Take $\mu \in \; ]-\inf \sigma (H^0),+\infty[$, and $\kappa \in [0,1]$. For any operator on $\cH$, we define the operator norm as
\begin{align*}
\norm{A} := \nor{A}{\cH \rightarrow \cH} = \mysup{\p \in \cH \\ \p \neq 0} \f{\norm{A \p}}{\norm{\p}},
\end{align*}
which can be equal to $+\infty$. The set $\{A : \cH \rightarrow \cH \;|\; \norm{A} < +\infty\}$ is usually denoted by $\cL(\cH)$. We define, for any operator $\Gamma$ of $\cH$,
\begin{align}\label{eq:norm_dmatrices}
\nor{\Gamma}{e} := \norm{ \pa{H^0 + \mu}^{\f \kappa 2} \; \Gamma \; \pa{H^0 + \mu}^{\f \kappa 2}}
\end{align}
which will be the norm on the space of density matrices. 
The norm on the space of Hamiltonians is the dual one, which is given by,
\begin{align}\label{eq:norm_params}
\nor{G}{a} := \norm{ \pa{H^0 + \mu}^{-\f \kappa 2} \; G \; \pa{H^0 + \mu}^{-\f \kappa 2}},
\end{align}
defined for any operator $G$ of $\cH$.
Taking $\kappa = 1$ is the most natural choice from a theoretical point of view since $\nor{\cdot}{e}$ becomes the natural energy norm in this case. 
However, the following proofs work for any $\kappa \in [0,1]$. 
In the numerical section, we will also take $\kappa = 1$, but one could consider $\kappa = 0$ to simplify the implementation.

\subsection{Resolvents and density matrices}%
Let $\cC$ be the contour introduced in Section~\ref{sub:non_deg} relying on Assumption~\ref{as:convgi} and take $G \in \Om$. 
Given the resolvent $\pa{z-H(G)}^{-1}$ of $H(G)$, we define
\begin{align}\label{eq:cauchy_formula}
\sol(G) := \f{1}{2\pi i} \oint_\cC  \pa{z-H(G)}^{-1} \d z.
\end{align}
By the spectral theorem, see~\cite[Section XII]{ReeSim4}, $\sol(G)$ is the spectral projection onto the one-dimensional space spanned by an eigenvector corresponding to the $k\thh$ eigenvalue of~$H(G)$.

\subsection{Pseudo-inverses}%
\label{sub:Pseudo-inverses}

We consider $G \in \Om$ and now introduce the pseudo-inverse operator
\begin{align}\label{def:pseudoinv}
K(G) := 
\left\{
\begin{array}{ll}
\pa{\pa{\lambda^k(G) - H(G)}_{\mkern 1mu \vrule height 2ex\mkern2mu \pa{\Ker \pa{\lambda^k(G) - H(G)}}^\perp}}^{-1} & \mbox{on }  \pa{\Ker \pa{\lambda^k(G) - H(G)}}^\perp, \\
0 & \mbox{on } \Ker \pa{\lambda^k(G) - H(G)},
\end{array}
\right.
\end{align}
extended on the whole of $\cH$ by linearity, that will be needed in the upcoming perturbation theory. We recall that $\Ker \pa{\lambda^k(G_j) - H(G_j)}$ is the one-dimensional vector space spanned by an eigenvector corresponding to the $k\thh$ eigenvalue of $H(G_j)$. \nv{Since $G \in \Om$, then $\nor{K(G)}{} \le a^{-1}$ ensuring that the operator $K(G)$ is not singular.{}}

Let us illustrate this definition for a finite-dimensional Hilbert space $\cH$. In this case, $H(G)$ is a hermitian
matrix, and taking $\pa{u_r(G)}_{r=1}^{\dim \cH}$ a basis of $\cH$ formed by eigenvectors of $H(G)$ sorted such that the corresponding eigenvalues $\pa{\lambda^r(G)}_{r=1}^{\dim \cH}$ are increasing, provides the explicit representation for the pseudo-inverse
\begin{align*}
K(G) = \sum_{\substack{1 \le r \le \dim \cH \\ r \neq k}}  \pa{\lambda^k(G) - \lambda^r(G)}^{-1} Q_{u_r(G)},
\end{align*}
where $Q_u$ is the projection operator on the vector space spanned by $u \in \cH$.

\subsection{Standard linear perturbation theory}%
\label{sub:Density matrix perturbation theory}

We now present standard perturbation theory, with the aim of comparing it with the proposed multipoint perturbation method. Let us take $G_1 \in \cG$, consider the non-degeneracy assumption~\ref{as:convgi} and take $G \in \Om$. We consider that the unperturbed operator is $H(G_1)$ i.e. we assume knowledge of $P_1 := \sol(G_1)$ and $K_1 := K(G_1)$ using~\eqref{def:pseudoinv}. Standard perturbation theory consists in deducing approximations of $\sol(G)$ in series of $G - G_1$. The base resolvent equation is

\begin{align}\label{eq:pert_thy}
\pa{z-H(G)}^{-1} = \pa{z-H(G_1)}^{-1} \pa{1 - (G-G_1) \pa{z-H(G_1)}^{-1}}^{-1}.
\end{align}
We then define for $\ell\in\N$,
\begin{align*}
\pt_\ell :=  \f{1}{2\pi i} \oint_\cC  (z-H(G_1))^{-1} \pa{(G-G_1)  (z-H(G_1))^{-1}}^\ell \d z, \qquad \bbP_\ell := \sum_{p=0}^{\ell}\pt_p,
\end{align*}
and recall the following classical result~\cite{Kato,ReeSim4}.

\begin{proposition}[Standard linear perturbation theory bound]\label{prop:standard_pt}
Take $k \in \N\backslash\{0\}$, $G_1, G \in \cG$ satisfying the non-degeneracy assumption~\ref{as:convgi}, consider the contour $\cC$ as in Section~\ref{sub:non_deg}, define $\gw := G-G_1$ and $C := \max_{z \in \cC} \nor{\pa{z-H(G_1)}^{-1}}{e}$, and assume that $\nor{\gw}{a} < 1/(2C)$. Then there exists a constant $c > 0$ independent of $\gw$ and $\ell$ such that for any $\ell \in \N$,
\begin{align}\label{eq:bound_standard_pert}
\nor{\sol(G) - \bbP_\ell}{e} \le  c \pa{C \nor{\gw}{a}}^{\ell+1}.
\end{align}
\end{proposition}

For the sake of completeness, we provide a proof in Section~\ref{sec:Proofs}. 
The terms $\bbP_\ell$ can be obtained explicitly using Cauchy's integration formula. For instance, the four first order terms read
\begin{align}\label{eq:dmpt}
\pt_0 &= P_1, \\
\pt_1 &= P_1 \gw K_1 + K_1 \gw P_1, \nonumber \\
\pt_2 &= \pa{P_1 \gw K_1 \gw K_1 + K_1 \gw P_1 \gw K_1 + K_1 \gw K_1 \gw P_1} \\
      &\qquad\qquad \qquad  - (P_1 \gw P_1 \gw K_1^2 + P_1 \gw K_1^2 \gw P_1 + K_1^2 \gw P_1 \gw P_1), \nonumber \\
\pt_3 &= (P_1 \gw P_1 \gw P_1 \gw K_1^3 + P_1 \gw P_1 \gw K_1^3 \gw P_1 + P_1 \gw K_1^3 \gw P_1 \gw P_1 + K_1^3 \gw P_1 \gw P_1 \gw P_1) \nonumber \\
&  + (P_1 \gw K_1 \gw K_1 \gw K_1 + K_1 \gw P_1 \gw K_1 \gw K_1 + K_1 \gw K_1 \gw P_1 \gw K_1 + K_1 \gw K_1 \gw K_1 \gw P_1) \nonumber \\
&  - \pa{P_1 \gw P_1 \gw K_1^2 \gw K_1 + P_1 \gw P_1 \gw K_1 \gw K_1^2 + K_1 \gw K_1^2 \gw P_1 \gw P_1 + K_1^2 \gw K_1 \gw P_1 \gw P_1 } \nonumber \\
&  - \pa{P_1 \gw K_1 \gw K_1^2 \gw P_1 + P_1 \gw K_1^2 \gw K_1 \gw P_1 + P_1 \gw K_1 \gw P_1 \gw K_1^2 + P_1 \gw K_1^2 \gw P_1 \gw K_1} \nonumber \\
&  -\pa{ K_1 \gw P_1 \gw K_1^2 \gw P_1 + K_1^2 \gw P_1 \gw K_1 \gw P_1 + K_1 \gw P_1 \gw P_1 \gw K_1^2 + K_1^2 \gw P_1 \gw P_1 \gw K_1}. \nonumber
\end{align}

\section{Main result}
\label{sec:main_sec}

In this section, we finally introduce the multipoint perturbation theory, which, compared to standard perturbation theory, allows one to retain more information of a series of density matrices related to $G_1,\ldots,G_n$.

\subsection{Main result}%

The following formula is the main result of this document. Its purpose is to express the resolvent $\pa{z - H(G)}^{-1}$ in terms of the known resolvents $\pa{z - H(G_j)}^{-1}$, so that the orthogonal projection $\sol(G)$ can also be expanded in terms of known terms
as will be shown in Section~\ref{sub:exp_F}.

\begin{theorem}[Multipoint resolvent formula]\label{thm:res_form}
Take $\bm{\alpha} = \pa{\alpha_j}_{j=1}^n \in \R^n$, $\bG = \pa{G_j}_{j=1}^n\in \cG^n$, $G \in \cG$, $z \in \C \backslash \pa{ \sigma(H(G)) \cup_{j=1}^n \sigma(H(G_j))}$ and $s_\bal :=\sum_{j=1}^{n} \alpha_j$ such that  $s_\bal \neq 0$. 
We define the operators
\begin{align*}
\bbH &:= s_\bal^{-1} \sum_{\substack{1 \le i < j \le n}} \alpha_i \alpha_j \pa{G_i - G_j} \pa{z - H(G_j)}^{-1} \pa{G_i - G_j} \pa{z - H(G_i)}^{-1}, \\
\bbA &:=s_\bal^{-1} \sum_{j=1}^{n} \alpha_j G_{j} \pa{z - H(G_j)}^{-1}, \qquad\qquad  \bbL := s_\bal^{-1}\sum_{j=1}^{n}  \alpha_j \pa{z - H(G_j)}^{-1}.
\end{align*}
We also define $\gw := G - \sum_{j=1}^{n} \alpha_j G_j$. If ${1 + \bbH - \Delm \bbA - \gw \bbL}$ is invertible, there holds
\begin{align}\label{eq:main_formula}
\boxed{\pa{z - H\pa{G}}^{-1} = \bbL \pa{1 + \bbH - \Delm \bbA - \gw \bbL}^{-1}.}
\end{align}

\end{theorem}
The proof is given in Section~\ref{sec:Proofs}. 

\begin{remark}[Linearity]
Formula~\eqref{eq:main_formula} provides an explicit formula for the resolvent of the Hamiltonian $H(G)$ that is expressed as a linear combination of the resolvents of $H(G_1),\ldots,H(G_n)$, namely $\bbL$, times a perturbation term when $\bbH - \Delm \bbA - \gw \bbL$ is small.
\end{remark}

\begin{remark}[$\bal$ is our degree of freedom]
Let us note that the degrees of freedom to use the formula are the coefficients $\alpha_j$ for $j$ from 1 to $n$, which can be freely chosen in $\R$ as long as $\sum_{j=1}^{n} \alpha_j \neq 0$.
\end{remark}

\begin{remark}[Term involving $\Del$]
We cannot ``absorb'' $\Delm \bbA$ in $-\gw \bbL$ in the sense that in general, $\Delm \bbA + \gw\bbL$ cannot be written as $\widetilde{\gw} \, \bbL$ for some operator $\widetilde{\gw}$, i.e., the two operators can be linearly independent. Hence we a priori need to keep both terms in the expansion. Note however that if we impose $s_\bal = 1$, then the term in $\Delm \bbA$ simply disappears.
\end{remark}

\begin{remark}[Reduction to standard linear perturbation theory]
	If we take $\alpha_1 = 1$ and $\alpha_j = 0$ for all $j \in \{2,\dots,n\}$, or $G_j = G_1$ for any $j \in \{2,\dots, n\}$, formula~\eqref{eq:main_formula} boils down to~\eqref{eq:pert_thy}. Indeed, in the latter case, $\bbH = 0$, $\bbA = G_1 \pa{z-H(G_1)}^{-1}$ and $\bbL = \pa{z-H(G_1)}^{-1}$ from which we easily obtain \eqref{eq:pert_thy}. Hence we recover standard perturbation theory.
\end{remark}


\subsection{Three perturbation parameters}

The next step of this multipoint linear perturbation theory is to expand the expression
\begin{align}\label{eq:neumann_three}
\pa{1 + \bbH - \Delm \bbA - \gw \bbL}^{-1}
\end{align}
appearing in~\eqref{eq:main_formula} in terms of a Neumann series, so the perturbative terms can be integrated over the complex contour $\cC$ and give approximations for the density matrix $\sol(G)$. For this, we need $\bbH - \Delm \bbA - \gw \bbL$ to be small enough. Let us therefore define the three smallness parameters
\begin{align*}
	\delta_{\bal,\bG} := \mymax{1 \le i,j \le n} \sqrt{|\alpha_i\alpha_j|} \nor{G_i - G_j}{a},\qquad \delta_\bal := \ab{1 - \sum_{j=1}^{n} \alpha_j}, \qquad \delta_\gw := \nor{\gw}{a},
\end{align*}
recalling $\gw = G - \sum_{j=1}^{n} \alpha_j G_j$, and where
\begin{itemize}
\item $\delta_{\bal,\bG}$ enters the estimation of $\bbH$, it is small when all the $G_j$'s are close to each others or when $\alpha_j$ is small,
\item $\delta_\bal$ enters into $\Delm \bbA$, it is small when $\sum_{j=1}^{n} \alpha_j$ is close to 1,
\item $\delta_\gw$ allows to estimate $\gw \bbL$, it is small when $G$ is close to $\sum_{j=1}^{n} \alpha_j G_j$.
\end{itemize}
All three parameters need to be small for the expansion of \eqref{eq:neumann_three} to be possible. In fact, a sufficient condition is that the $G_j$'s are close to each other and $G$ is close to the affine space
\begin{align*}
\Aff \pa{G_j}_{j=1}^{n}:= \acs{\sum_{j=1}^{n} \alpha_j G_j \;\middle|\; \alpha_j \in \R, \sum_{j=1}^{n} \alpha_j = 1},
\end{align*}
which is the smallest affine subspace of $\cG$ containing all the $G_j$'s.

\subsection{Perturbation bound}%
\label{sub:Perturbation bound}
The final step consists in integrating the multipoint perturbative resolvent formulation along the complex contour $\cC$ to obtain corresponding approximations for the density matrix $\sol(G)$ via Cauchy's formula~\eqref{eq:cauchy_formula}. We first define the perturbation terms
\begin{align}
\label{eq:Dabc}
\cD_{(2a,b,c)} &:= 
\f{(-1)^{a}}{2\pi i}(s_\bal -1)^b \oint_\cC \bbL \sum_{\substack{B_1, \dots, B_{a+b+c} \in \{\bbH, \bbA, \gw \bbL\} \\ \ab{\{i \;|\; B_i = \bbH\}} = a \\ \ab{\{i \;|\; B_i =  \bbA\}} = b\\ \ab{\{i \;|\; B_i = \gw\bbL\}} = c}} 
B_1 B_2 \cdots B_{a+b+c} \d z, \nonumber\\
\cD_{(2a+1,b,c)} &:= 0,\qquad  \qquad \cD_{\ell} :=  \sum_{\substack{a,b,c \in \N \\ a+b+c = \ell}} \cD_{(a,b,c)}, 
\qquad \qquad
\D_{\ell} := \sum_{p=0}^\ell \cD_{p},
\end{align}
In the summation formula for $\cD_{(2a,b,c)}$, we consider all the combinations of products $B_p$'s in $\{\bbH, \bbA, \gw \bbL\}$, each once, with the convention that the empty sum on the first line is taken as the identity ; to clarify how to compute them, particular examples will be presented in the upcoming Section~\ref{sub:exp_F}. Then we have the following result.
\begin{corollary}[Multipoint perturbation bound]\label{cor:multipoint_bound}
 Let us take $k \in \N\backslash\{0\}$, $(G_j)_{j=1}^n \in \cG^n$ and consider that Assumption~\ref{as:convgi} is satisfied. Let us define 
\begin{align*}
        c_\bal := s_\bal^{-1} \pa{1 + \max_{1 \le j \le n} \ab{\alpha_j}}
\end{align*}
and take $\ell \in \N$. 
 There exist $C_\ell \in \R_+$ and $m_\ell \in \N$ such that for any $G \in \Om$ and any $\bal \in \R^n$ with $s_\bal \neq 0$ and $\delta_{\bal,\bG} + \delta_\bal + \delta_\gw \le 1/2$,
\begin{align}\label{eq:main_bound}
	\nor{\sol(G) - \D_\ell}{e} \le C_\ell c_\bal^{m_\ell} \pa{\delta_{\bal,\bG}^{\ell +1 + \xi_\ell} + \delta_\bal^{\ell +1} + \delta_\gw^{\ell+1}},
\end{align}
where $\xi_\ell$ is $1$ if $\ell$ is even and $0$ otherwise, and $C_\ell $ and $m_\ell$ are independent of $G$ and $\bal$ but depend on $\ell$.
\end{corollary}
We give a proof in Section~\ref{sec:Proofs}. The main property of~\eqref{eq:main_bound} is that the power of $\delta_{\bal,\bG}$ is always even and one extra order of convergence is gained for $\ell$ even.
As we will see in Section~\ref{ssub:Comparision to standard perturbation theory}, when $\cH$ is finite-dimensional the terms $\D_\ell$ and $\bbP_\ell = \sum_{p=0}^{\ell} \pt_p$ (which appears in standard perturbation theory) have the same computational cost in terms of numbers of matrix products. In consequence, $\D_\ell$ can be directly compared with $\bbP_\ell$ to observe in which regimes multipoint perturbation theory can lead to more accurate approximations than standard perturbation theory.

\subsection{Case $\delta_\gw = 0$ and $\delta_\bal = 0$}\label{sub:case_efficient}
A particularly interesting situation is when $G \in \Aff \pa{G_j}_{j=1}^n$ and $\delta_\bal = 0$. Then we see in~\eqref{eq:main_bound} that for $\ell \in 2\N$, the error is of order $\ell +2$ in $\delta_{\bal,\bG}$ while the error in~\eqref{eq:bound_standard_pert} is of order $\ell +1$ in $G - G_1$. We thus gain one order of convergence. This is hence in the neighborhood of $\Aff \pa{G_j}_{j=1}^n$ and when $\delta_{\bal,\bG}$ is small that we can expect multipoint perturbation to be particularly performant. We will  illustrate this numerically in Section~\ref{sub:all_zero}.

\begin{remark}[Universality of multipoint approximation at zeroth order]
In the case of zeroth order multipoint approximation, a similar but different bound holds for any mapping and no special resolvent identity, particular to eigenvalue problems, is required. In this general case as well, the first order disappears as we explain here.

Consider $\cM$ and $\cN$ two normed vector spaces (the first one being real), with norms denoted by $\nor{\cdot}{\cM} $ and $\nor{\cdot}{\cN}$, and a $\cC^2$ map $f : \cM \rightarrow \cN$. Take some bounded subset $\cB \subset \cM$, there exists $C > 0$ such that for any $y,s \in \cB$,
\begin{equation}
	\label{eq:taylor}
	\nor{f(s) -  f(y) -  \pa{\d_y f}(s-y)}{\cN} \le C \nor{s-y}{\cM},
\end{equation}
where $\d_y f$ denotes the differential of $f$ at $s$. Take $x_1, \dots,x_n \in \cB$ and $\alpha_1,\dots \alpha_n \in \R$ such that $\sum_{j=1}^{n} \alpha_j = 1$, define $x := \sum_{j=1}^{n} \alpha_j x_j$, and assume that $x \in \cB$. We have then
\begin{align*}
&f(x) - \sum_{j=1}^{n} \alpha_j f(x_j) = \sum_{j=1}^{n} \alpha_j \pa{f(x)  + \pa{\d_{x}f}(x_j - x)- f(x_j)}.
\end{align*}
By defining $\delta_{\bal,\bx} := \max_{j=1}^{n}\sqrt{|\alpha_j|} \, \nor{x-x_j}{\cM} \le n \max_{1 \le i,j \le n} \sqrt{|\alpha_j|} \ab{\alpha_i} \, \nor{x_i-x_j}{\cM}$, and using~\eqref{eq:taylor}, this yields the estimate
\begin{align*}
	\nor{f(x) - \sum_{j=1}^{n} \alpha_j f(x_j)}{\cN} &\le  C \sum_{j=1}^{n} \ab{\alpha_j} \nor{x - x_j}{\cM}^2  
	\le  n C \delta_{\bal,\bx}^2.
\end{align*}
\end{remark}

\subsection{Exploiting the symmetry group of $H^0$}\label{sub:symmetry_gp}

If the Hamiltonian $H^0$ is invariant with respect to some symmetry, one can exploit it to
approximate density matrices for many more admissible $G$'s, without additional knowledge. The corresponding principle is explained in the following. We introduce the symmetry group of $H^0$,
\begin{align*}
\cU := \acs{U \text{ bounded linear operator of } \cH, \; U^* U = \1,\; [H^0,U] = 0}.
\end{align*}
For instance if $H^0 = -\Delta$, $\cU$ is the Galilean group, i.e. the group spanned by translations and rotations. For any $U \in \cU$, and any $G \in \cG$, we have
\begin{align*}
U \pa{z-H(G)}^{-1} U^* = \pa{z-U H(G) U^*}^{-1}= \pa{z-H( UGU^*)}^{-1},
\end{align*}
so by integration, $U \sol(G) U^* = \sol(UGU^*)$. Hence it is natural to use this information to extend the multipoint perturbation formula~\eqref{eq:main_formula}.
Thus, denoting $\bU =(U_1,\ldots,U_n)$ with $U_i \in \cU$, 
defining
\begin{align*}
	\bbzH &:= \sum_{\substack{1 \le i < j \le n}} \alpha_i \alpha_j \pa{U_i G_i U_i^* - U_j G_j U_j^*}  U_j \pa{z - H(G_j)}^{-1} U_j^*\\
		  & \qquad \qquad \qquad \qquad \qquad \qquad \times \pa{U_i G_i U_i^* - U_j G_j U_j^*} U_i \pa{z - H(G_i)}^{-1} U_i^*, \\
\bbzA &:=s_\bal^{-1} \sum_{j=1}^{n} \alpha_j U_j G_{j} \pa{z - H(G_j)}^{-1} U_j^* , \quad \bbzL := s_\bal^{-1}\sum_{j=1}^{n}  \alpha_j U_j \pa{z - H(G_j)}^{-1} U_j^*,
\end{align*}
and $\gw_\bU := G - \sum_{j=1}^{n} \alpha_j U_j G_j U_j^*$, equation~\eqref{eq:main_formula} becomes
\begin{align}\label{eq:main_formula-2}
\pa{z - H\pa{G}}^{-1} = \bbzL \pa{1 + \bbzH - \Delm \bbzA - \gw_\bU \bbzL}^{-1}.
\end{align}
From this, the results obtained in~\eqref{eq:main_bound} can easily be extended to this setting as well, and
\begin{align*}
\nor{\sol(G) - \D_\ell^\bU}{e} \le C_\ell c_\bal^{m_\ell} \pa{\delta_{\bal,\bU\bG\bU^*}^{\ell +1 + \xi_\ell} + \delta_\bal^{\ell +1} + \delta_{\gw_\bU}^{\ell+1}},
\end{align*}
where ${\bU \bG \bU^* := \pa{U_j G_j U_j^*}_{j=1}^n}$ and where $\D_\ell^\bU$ is $\D_\ell$ under the above transformations, which can be recast as described in the upcoming~\eqref{eq:transfo}. Minimizing the previous bound over $\bal \in \R^n \backslash \{0\}$ and $\bU \in \cU^n$ provides a better estimate of $ \sol(G)$.

\section{Explicit computation of the first order terms}%
\label{sec:Firsts order}
In this section, we first present the first orders of $\pa{z-H(G)}^{-1}$ using the main formula~\eqref{eq:main_formula} and then the first orders of $\sol(G)$.

\subsection{Preliminary details}%
\label{sub:prelim}
As in Section~\ref{sub:Pseudo-inverses}, we denote by $\lambda^a(G)$ the $a\thh$ eigenvalue of $H(G)$, for any $G \in \cG$ and any $a \in \N\backslash\{0\}$. We start by defining the pseudo-inverses
\begin{align*}
K_{ij} := 
\left\{
\begin{array}{ll}
\pa{\pa{\lambda^k(G_i) - H(G_j)}_{\mkern 1mu \vrule height 2ex\mkern2mu \pa{\Ker\pa{\lambda^k(G_j) - H(G_j)} }^\perp}}^{-1} & \mbox{on }  \pa{\Ker \pa{\lambda^k(G_j) - H(G_j)}}^\perp\\
0 & \mbox{on }\Ker \pa{\lambda^k(G_j) - H(G_j)},
\end{array}
\right.
\end{align*}
extended to $\cH$ by linearity, which are going to be needed in the context of multipoint perturbation theory, which are well-defined if Assumption~\ref{as:convgi} holds. 
In the case of a finite-dimensional Hilbert space $\cH$, and with the same notation as in Section~\ref{sub:Pseudo-inverses}, there holds
\begin{align*}
K_{ij} = \sum_{\substack{1 \le a \le \dim \cH \\ a \neq k}}  \pa{\lambda^k(G_i) - \lambda^a(G_j)}^{-1} Q_{u_a(G_j)}.
\end{align*}
This is the form that can be used in practical implementations, if all the $u_a(G_j)$'s were stored during the step of finding the eigenmodes of $H(G_j)$ and building the $K_j$'s. Another way of applying $K_{ij}$ or $K_i$ to a vector is to solve a linear equation. 
Thus, in practice, the cost of computing $K_{ij}$ either via a direct method or by applying it to a vector is comparable to the computation of $K_i$, provided it is done using the same method.
Furthermore, we propose in Section~\ref{sub:Kij_from_Ki} in the Appendix a way to approximate the $K_{ij}$'s from the $K_i$'s.

Then, we define, for $A,B \in \cG$, and $a,b,c\in \{1,\ldots, n\}$ the following auxiliary quantities
\begin{align*}
	\cI_{a,b}(A) &:= \f{1}{2\pi i} \oint_\cC \pa{z-H(G_a)}^{-1} A \pa{z-H(G_b)}^{-1}  \d z, \\
\cI_{a,b,c}(A,B) &:= \f{1}{2\pi i} \oint_\cC  \pa{z-H(G_a)}^{-1} A \pa{z-H(G_b)}^{-1}  B \pa{z-H(G_c)}^{-1}\d z,
\end{align*}
and compute their explicit form.
\begin{proposition}\label{prop:integrals}
For $A,B \in \cG$, and $a,b,c\in \{1,\ldots, n\}$, and defining $P_j := \sol(G_j)$, we have
\begin{align}\label{eq:comp_cI}
\cI_{a,b}(A) &=P_a A K_{ab} + K_{ba} A P_b, \\
\cI_{a,b,c}(A,B) &= P_a A K_{ab} B K_{ac} + K_{ba} A P_b B K_{bc} + K_{ca} A K_{cb} B P_c  \nonumber \\
& \qquad \qquad - P_a A P_b B K_{ac} K_{bc} - P_a A K_{ab} K_{cb} B P_c - K_{ba} K_{ca} A P_b B P_c. \nonumber
\end{align}
\end{proposition}
We provide a proof in Section~\ref{sec:Proofs}.

\subsection{First terms in the expansion of $\sol(G)$}\label{sub:exp_F}

We now compute the terms in~\eqref{eq:Dabc} up to order 2,
\begin{align}
\cD_{(0,0,0)} =\f{1}{2\pi i} \oint_\cC \bbL \d z =s_\bal^{-1} \sum_{j=1}^{n} \alpha_j \sol(G_j),
\end{align}
then $\cD_{(1,0,0)}=0$,
\begin{align*}
\cD_{(2,0,0)} = -\f{1}{2\pi i} \oint_\cC \bbL \bbH \d z = -s_\bal^{-2}\sum_{\substack{1 \le a < b \le n \\ 1 \le j \le n}} \alpha_j \alpha_a \alpha_b \cI_{j,b,a}(G_{ab},G_{ab}).
\end{align*}
Similarly, we have
\begin{align*}
	\cD_{(0,1,0)} &= \f{1}{2\pi i} \oint_\cC \bbL \gw \bbL \d z 
=s_\bal^{-2}\sum_{\substack{1 \le i,j \le n }} \alpha_i\alpha_j \cI_{i,j}(\gw)\\
\cD_{(0,0,1)} &= \f{\Delm}{2\pi i} \oint_\cC \bbL \bbA \d z 
=\Delm s_\bal^{-2}\sum_{\substack{1 \le i,j \le n}} \alpha_i\alpha_j \cI_{i,j}(G_j) \\ 
\cD_{(0,2,0)} &=  \f{1}{2\pi i} \oint_\cC \bbL \gw \bbL \gw \bbL \d z 
 =s_\bal^{-3}\sum_{\substack{1 \le i,j,k \le n}}  \alpha_i\alpha_j  \alpha_k \cI_{i,j,k}(\gw,\gw)  \\
\cD_{(0,0,2)} &= \Delm^2 \f{1}{2\pi i} \oint_\cC \bbL \bbA^2 \d z 
 =\Delm^2s_\bal^{-3}\sum_{\substack{1 \le i,j,k \le n}}  \alpha_i\alpha_j \alpha_k \cI_{i,j,k}(G_j,G_k) \\
\cD_{(0,1,1)} &=  \f{\Delm}{2\pi i} \oint_\cC \bbL \pa{\bbA \gw \bbL + \gw \bbL \bbA}  \d z \\
	     & = \Delm s_\bal^{-3}\sum_{\substack{1 \le i,j,k \le n}}  \alpha_i\alpha_j\alpha_k \pa{ \cI_{i,j,k}\pa{G_j,\gw} + \cI_{i,j,k}(\gw,G_k)}.
\end{align*}
Finally, we have 
\begin{align*}
	\cD_0 &= \cD_{\pa{0,0,0}}, \\
\cD_1 &= \cD_{\pa{0,1,0}} + \cD_{\pa{0,0,1}}\\
\cD_2 &= \cD_{(2,0,0)} + \cD_{(0,2,0)} + \cD_{(0,0,2)} + \cD_{(0,1,1)},
\end{align*}
and
\begin{align*}
\D_0 &= \cD_0, \qquad \D_1 = \cD_0 + \cD_1, \qquad \D_2 = \cD_0 + \cD_1 + \cD_2.
\end{align*}
The terms $\D_\ell$ for odd $\ell$ are not of primary practical interest because they do not improve the convergence rate of $\delta_{\bal,\bG}$ with respect to $\D_{\ell -1}$, as it is reflected in~\eqref{eq:main_formula}.


\subsection{Exploiting the symmetry group of $H^0$}

Recalling Section~\ref{sub:symmetry_gp} where we presented how to use the symmetry group of $H^0$, if we want to deduce the expansion of $\sol(G)$ as in Section~\ref{sub:exp_F}, we just need to change 
\begin{align}\label{eq:transfo}
G_j \rightarrow U_j G_j U_j^*, \qquad \qquad P_j \rightarrow U_j P_j U_j^*,  \qquad \qquad K_{ij} \rightarrow U_j K_{ij} U_j^*.
\end{align}
Moreover, we remark that $\Om$ remains unchanged.


\subsection{Complexity analysis}%
\label{sub:Complexity analysis}

In this section, we compare the complexity of the standard perturbation approximation and the multipoint perturbation method depending on the order of the expansion $\ell$ and the number of points $n$.
We denote by 
\begin{itemize}
	\item $m$ the cost of applying an operator of the kind $\left| \vp \right> \left< \p \right|$ to a vector, where $\vp, \p \in \cH$,
	\item $q$ the cost of applying $K_{ab}$ or $K_a$ to a vector in $\cH$,
	\item $p$ the cost of applying $\gw$ or $G_i$ to a vector in $\cH$.
\end{itemize}

Of course, this is intended for the case of a finite-dimensional Hilbert space. Further, the complexity of $q,p$ with respect to the dimension of the Hilbert space very much depends on the properties of the induced matrices. 
If the discrete setting results from a discretization of a problem posed in an infinite-dimensional Hilbert space, then these properties also depend much on the type of discretization method that is employed. We keep it therefore abstract in this analysis such that it can be assessed for each method individually.

\subsubsection{Using the symmetries of the $\cI$'s}%
\label{ssub:sym_Js}

To decrease the number of operations, we first exploit the symmetries of $\cI_{i,j}$ and $\cI_{i,j,k}$ and write
\begin{align*}
	\cD_{(0,1,0)} 
	= 2! s_\bal^{-2}\sum_{\substack{1 \le i < j \le n }} \alpha_i\alpha_j \cI_{i,j}(\gw) +  s_\bal^{-2}\sum_{i=1}^{n}  \alpha_i^2 \cI_{i,i}(\gw),
\end{align*}
reducing to $J_1 := \f 12 n (n+1)$ the number of operations in the two sums instead of $n^2$. Similarly,
\begin{multline*}
\cD_{(0,2,0)}=
3!  \sum_{\substack{1 \le i < j < k \le n}}  \alpha_i\alpha_j  \alpha_k \cI_{i,j,k}(\gw,\gw) + 2! \sum_{\substack{1 \le i < k \le n}}  \alpha_i^2  \alpha_k \cI_{i,i,k}(\gw,\gw) \\
 + 2! \sum_{\substack{1 \le i < j \le n}}  \alpha_i^2  \alpha_j \cI_{i,j,i}(\gw,\gw)+ 2! \sum_{\substack{1 \le i < j \le n}}  \alpha_i \alpha_j^2  \cI_{i,j,j}(\gw,\gw) +  \sum_{i=1}^{n}\alpha_i^3  \cI_{i,i,i}(\gw,\gw).
\end{multline*}
Using that 
\begin{align*}
    \sum_{\substack{1 \le i < j < k \le n}} 1= \sum_{k=1}^{n} \sum_{j=1}^{k-1} \sum_{i=1}^{j-1} 1 = \f{n}{6}\pa{n^2 -3n +4}
\end{align*}
brings the number of operations from $n^3$ to 
\begin{align*}
	J_2 := \f 16 n (n^2 - 3n +4) + 3 \times 2 \times \f 12 n(n-1) + n = \f 16 n \pa{n^2 - 15n - 8}.
\end{align*}
Note that the leading order remains the same but the preconstant can be significantly reduced. 

\subsubsection{Costs of intermediate quantities}%
\label{ssub:Costs of intermediate quantities}

For any operator $L$, we denote by $\mathfrak{C}(L) \in \N^2$ 
the cost of applying $L$ to a vector, the first component being the offline cost and the second one being the online cost. The offline computations are the ones that can be performed once the $G_j$'s are fixed and that do not depend on $G$. The online computations are the remaining ones that can only be performed once a $G$ is chosen. 
Note that our aim here is to apply $\bbP_\ell$ and $\D_\ell$ to a vector. 
We separate the offline and online computations so that we can assess the complexity for both cases when the perturbation is applied to only one $G$ as well as if it is applied for many different $G$'s, such as in a many-query context, where the $G_j$'s are fixed once and for all and the computations need then to be done for many sample points $G$'s. 
Let us start with a typical example. For $a,b,j = 1,\ldots, n$, to apply $P_a G_j K_{ab}$ to a vector $\p\in\mathcal H$, we write
\begin{align*}
P_a G_j K_{ab} \p = \ps{u_k(G_a), G_j K_{ab} \p} u_k(G_a) =  \ps{K_{ab} G_j  u_k(G_a),   \p} u_k(G_a)
\end{align*}
so we can pre-compute $\vp_{a,b,j} := K_{ab} G_j  u_k(G_a)$ offline which costs $p + q$. Finally, one can compute $\ps{\vp,\p}$ online, this costs $m$, and we write $\co\pa{P_a G_j K_{ab}} = (p+q,m)$. Therefore we obtain for the following operators
\begin{align*}
	\co\pa{P_a G_j K_{ab}} &=\co\pa{K_{ab}G_j P_a } =  (p+q,m) \\
	\co\pa{P_a \gw K_{ab}} &= \co\pa{K_{ab} \gw P_a} = (0,p+q+m) \\
	\co\pa{\cI_{a,b}(G_j)} &= 2(p+q,m)\\
	\co\pa{\cI_{a,b}(\gw)} &= 2(0,p+q+m),
\end{align*}
using~\eqref{eq:comp_cI} for the expression of $\cI_{a,b}$.
We now compute the cost of the $\cI_{a,b,c}$'s. For instance, for the term $P_a \gw P_b G_j K_{ac} K_{bc}$, we have the decomposition
\begin{align*}
	P_a \gw P_b G_j K_{ac} K_{bc} \p = u_k(G_a) \ps{\gw u_k(G_a),u_k(G_b)} \ps{K_{bc} K_{ac} G_j u_k(G_b),\p},
\end{align*}
so that
\[
    \co\pa{P_a \gw P_b G_j K_{ac} K_{bc}} = \pa{2q+p,2m+p}.
\]
Similarly, we obtain
\begin{align*}
\co\pa{P_a \gw K_{ab} \gw K_{ac}} &= (0,m+2q+2p) \\
\co\pa{P_a \gw P_b \gw K_{ac} K_{bc}} &= (0,2(m+q+p)) \\
\co\pa{\cI_{a,b,c}(\gw,\gw)} &= (0,3(m+2q+2p) + 6(m+q+p)) = 3(0,3m+4q+4p),
\end{align*}
then
\begin{align*}
\co\pa{P_a G_j K_{ab} \gw K_{ac}} &= \co\pa{K_{ba} G_j P_b \gw K_{bc}} = (p+q,m+p+q) \\
\co\pa{K_{ca} G_j K_{cb} \gw P_c } &= (0,m+2p + 2q) \\
\co\pa{P_a G_j P_b \gw K_{ac} K_{bc}} &= (m+p,m+p+2q) \\
\co\pa{P_a G_j K_{ab} K_{cb} \gw P_c} &= (p+2q,2m+p)  \\
\co\pa{K_{ba} K_{ca} G_j P_b \gw P_c} &= (p+2q,2m+p)  \\
\co\pa{\cI_{a,b,c}(G_j,\gw)} &= (m+5p+6q,8m+7p+6q),
\end{align*}
as well as
\begin{align*}
\co\pa{P_a \gw K_{ab} G_j K_{ac}} &= (0,m+2p+2q) \\
\co\pa{K_{ba} \gw P_b G_j K_{bc}} &=\co\pa{K_{ca} \gw K_{cb} G_j P_c } = (p+q,m+p+q) \\
\co\pa{P_a \gw P_b G_j K_{ac} K_{bc}} &= (p+2q,2m+p) \\
\co\pa{P_a \gw K_{ab} K_{cb} G_j P_c} &= (p+2q,2m+p)  \\
\co\pa{K_{ba} K_{ca} \gw P_b G_j P_c} &= (m+p,m+p+2q) \\
\co\pa{\cI_{a,b,c}(\gw,G_j)} &= (m+5p+6q,8m+7p+6q),
\end{align*}
and
\begin{align*}
\co\pa{P_a G_i K_{ab} G_j K_{ac}} &= \co\pa{K_{ba} G_i P_b G_j K_{bc}} =\co\pa{K_{ca} G_i K_{cb} G_j P_c } = (2p+2q,m) \\
\co\pa{P_a G_i P_b G_j K_{ac} K_{bc}} &= \co\pa{P_a G_i K_{ab} K_{cb} G_j P_c} =\co\pa{K_{ba} K_{ca} G_i P_b G_j P_c} = (m+2p+2q,m)  \\
\co\pa{\cI_{a,b,c}(G_i,G_j)} &= \pa{3m+12p + 12q,6m}.
\end{align*}

\subsubsection{Final costs}%
\label{ssub:Final costs}

With these different intermediate calculations, we are now ready to estimate the cost of computing the standard perturbation terms $\bbP_j$'s at few first orders. For this, using~\eqref{eq:dmpt}, and noting that no precomputation can be performed since the calculations involve $\gw = G-G_1$ which changes with the potential $G$, there holds
\begin{align*}
        \co\pa{\pt_0} &= (0,m),  \\
	\co\pa{\pt_1} &= 2(0,m+p+q), 
	\\ 
	\co\pa{\pt_2} &= 3(0,3m+4p+4q), \\
        \co\pa{\pt_3} &= 20(0,2m+3p+3q).
\end{align*}
Therefore we obtain for the standard perturbation method
\begin{align*}
	\co\pa{\bbP_0} &= \co\pa{\pt_0} = (0,m) \\
	\co\pa{\bbP_1} &= \co\pa{\bbP_0} + \co\pa{\pt_1} = (0,3m+2p+2q) \\
	\co\pa{\bbP_2} &= \co\pa{\bbP_1} + \co\pa{\pt_2} = 2(0,6m+7p+7q) \\
		\co\pa{\bbP_3} &= \co\pa{\bbP_2} + \co\pa{\pt_3} = 2(0,26m+37p+37q).
\end{align*}
Now for the multipoint perturbation method, we start by giving the cost estimation for the $\cD_{(a,b,c)}$'s using the explicit expressions provided in Section~\ref{sub:exp_F}, as well as the computations from Section~\ref{ssub:Costs of intermediate quantities} combined with Section~\ref{ssub:sym_Js}, which are, up to second order
\begin{align*}
\co\pa{\cD_{(0,0,0)}} &=  n (0,m)  \\
	\co\pa{\cD_{(0,1,0)}} &= J_1 \co\pa{\cI_{a,b}(\gw)} =  n(n+1)(0,p+q+m)\\
	\co\pa{\cD_{(0,0,1)}} &= n^2 \co\pa{\cI_{a,b}(G_j)} = 2n^2(p+q,m)\\
	\co\pa{\cD_{(2,0,0)}} &= n^3 \co\pa{\cI_{a,b,c}(G_i,G_j)} = n^3 \pa{3m+12p + 12q,6m} \\
	\co\pa{\cD_{(0,2,0)}} &= J_2 \co\pa{\cI_{a,b,c}(\gw,\gw)} = \frac{1}{2} n \pa{n^2 - 15n - 8}(0,3m+4q+4p) \\
	\co\pa{\cD_{(0,0,2)}} &=  n^3 \co\pa{\cI_{a,b,c}(G_i,G_j)} = 
    3 n^3 (m+4p+4q,2m)
 \\
	\co\pa{\cD_{(0,1,1)}} &=  2n^3 \co\pa{\cI_{a,b,c}(G_i,\gw)} = 2n^3 (m+5p+6q,8m+7p+6q).
\end{align*}
Therefore, we obtain for the multipoint perturbation method up to second-order
\begin{align*}
	\co\pa{\D_0} &= n (0,m) \\
	\co\pa{\D_1} &= \co\pa{\D_0} + \delta_{\gw \neq 0} \co\pa{\cD_{(0,1,0)}} + \delta_{s_\bal \neq 1} \co\pa{\cD_{(0,0,1)}} \\
	\co\pa{\D_2} &= \co\pa{\D_1} + \co\pa{\cD_{(2,0,0)}} + \delta_{\gw \neq 0} \co\pa{\cD_{(0,2,0)}} + \delta_{s_\bal \neq 1} \co\pa{\cD_{(0,0,2)}} + \delta_{\gw \neq 0}\delta_{s_\bal \neq 1} \co\pa{\cD_{(0,1,1)}}.
\end{align*}
To lower again the computational cost of terms, one can use multipoint theory on $\Aff (G_j)_{j=1}^n$ and then standard perturbation until $G$, giving $\delta_{\gw} = 0$ for the first step. In this case, if we also use that $s_\bal = 1$, the computational cost of computing $\D_2$ is $\co\pa{\cD_{(2,0,0)}} = 3 n^3 \pa{m+4p +4q,2m}$ at leading order.
In order to reduce $n$, one could also consider to take only the $G_j$'s which are closest to~$G$.

\subsubsection{Comparison to standard perturbation theory}%
\label{ssub:Comparision to standard perturbation theory}

We now summarize the computational complexity and approximation order of the standard and multipoint perturbation theories in the following tables~\ref{tab:standard_pert} and~\ref{tab:multipoint}, 
the latter being taken in the regime where ${\delta_\alpha=\delta_\gw=0}$ since it is the most efficient scenario for this method.
\begin{table}[h]
    \centering
\begin{center}
    \begin{tabular}{|r|c|c|c|c|}
    \hline
        Expansion order $\ell$ & 0 & 1 & 2 & 3  \\ \hline
        Offline complexity & 0 & 0 & 0 & 0 \\ \hline
        Online complexity & $m$ & $2m+2p+2q$ & $12m+14p+14q$ & $52m+74p+74q$  \\ \hline
        Approximation order & 1 & 2 & 3 & 4 \\ \hline 
    \end{tabular}
\end{center}
\caption{Computational cost for standard perturbation theory}
\label{tab:standard_pert}
\end{table}

\begin{table}[h]
    \centering
\begin{center}
    \begin{tabular}{|r|c|c|c|}
        \hline
        Expansion order $\ell$ & 0 & 1 & 2  \\ \hline
        Offline complexity & 0 & 0 & $n^3(3m+12p+12q)$  \\ \hline
        Online complexity & $nm$ & $nm$ & $nm(1+6n^2)$  \\ \hline
        Approximation order & 2 & 2 & 4 \\ \hline 
    \end{tabular}
\end{center}
    \caption{Multipoint perturbation theory with $(\delta_\alpha=\delta_\gw=0)$}
    \label{tab:multipoint}
\end{table}


Let us now compare the approximations at second order within the framework of the two theories. This corresponds to the second column of Table~\ref{tab:standard_pert} for the standard perturbation theory and the first column of Table~\ref{tab:multipoint} for the multipoint perturbation theory.
In both cases, the offline cost is zero. The online costs are respectively $3m+2p+2q$ and $nm$ respectively. Therefore, as long as $nm \le 3m+2p+2q $, multipoint perturbation theory is more efficient than standard perturbation theory, note however that the multipoint perturbation theory is restricted here to the case $\delta_\alpha=\delta_\gw=0$.

Similarly, at fourth order, we need to look at the fourth column of Table~\ref{tab:standard_pert} and third column of Table~\ref{tab:multipoint}, and we observe that multipoint perturbation is more efficient than standard perturbation when 
$nm(1+6n^2) \le 52m+74p+74q$
which in particular holds true when $p$ and $q$ are large.


\section{Numerical examples}
\label{sec:simus}

In this section we apply the multipoint perturbation method to Schrödinger operators. We then observe in which domain of $\bal$ and $\gw$ it is efficient. Finally, we conclude by presenting test cases where multipoint perturbation is more efficient than standard perturbation.

We consider a spatial domain $\Omega := [-\pi,\pi[$, and the Hilbert space will be the set of one-dimensional periodic functions ${\cH = L^2\ind{per}(\Omega)}$, ${H^0 = -\Delta}$, and we will define the admissible set $\cG$ as the set of multiplication operators by smooth potentials $V \in \cC^{\infty}\ind{per}(\Omega)$. We use a planewaves basis discretization 
\[
    \cH_M := \text{Span}\acs{ x\mapsto e^{i x m}, \; m\in\Z, \;  |m| \le \floor{M/2}}.
\]
The discretization parameter $M$ is taken to be $30$ and fixed throughout the numerical section.
In Section~\ref{sub:Convergence of the quantities} of the appendix, we show that this value of $M$ is large enough  so that the studied quantities can be considered as converged with respect to $M$.

We define the following norm and corresponding relative distances
\begin{align*}
	\nor{A}{2} &:= \sqrt{\tr A^* A} \\
	d_\kappa(A,B) &:= \f{2\nor{\pa{H^0 + 1}^{\f \kappa 2} \pa{A-B} \pa{H^0 + 1}^{\f \kappa 2}}{2}}{\nor{\pa{H^0 + 1}^{\f \kappa 2} A \pa{H^0 + 1}^{\f \kappa 2}}{2}+ \nor{\pa{H^0 + 1}^{\f \kappa 2} B \pa{H^0 + 1}^{\f \kappa 2}}{2}} \\
	D_e(A,B) &:= d_1(A,B), \\
        D_a(A,B) &:= d_{-1}(A,B)
\end{align*}
where $A$ and $B$ are complex square matrices. Thus $D_e$ is the relative distance in energy norm and $D_a$ in dual norm. We use $\nor{\cdot}{2}$ instead of the supremum norm $\nor{\cdot}{\cH \rightarrow \cH}$ because in finite dimension they are equivalent, and $\nor{\cdot}{2}$ is numerically cheaper and simpler to evaluate.

We define the potentials $V_1$ and $V_2$ as periodized Gaussian functions, and $V_3$ and $V_4$ as superposition of $\cos$ and $\sin$ functions as represented in Figure~\ref{fig:pots}.


\begin{figure}[h]
\begin{center}
\includegraphics[width=12cm,trim={0cm 0cm 0cm 0cm},clip]{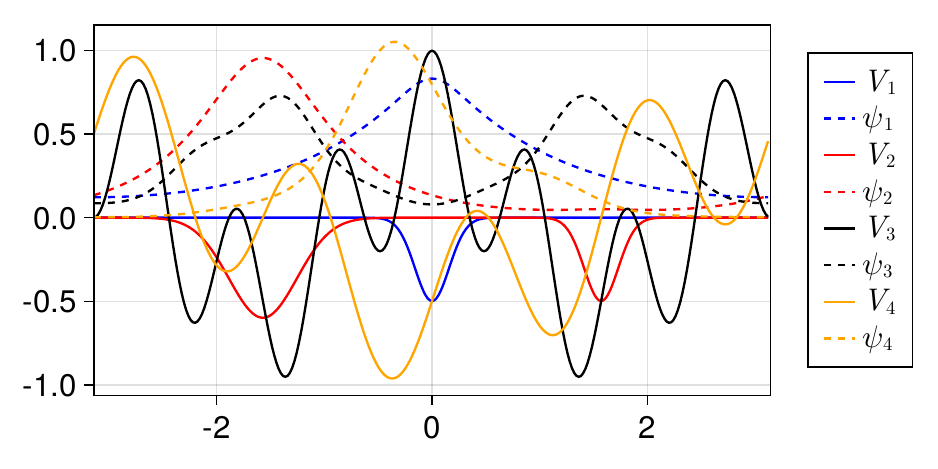} 
\caption{Potentials $V_j$, $j \in \{1,2,3,4\}$ together with the ground states $\psi_j$ of $H^0+V_j$.}\label{fig:pots}
\end{center}
\end{figure}

Before starting the presentation of the numerical results, let us quickly recall the definitions of the three smallness parameters
\begin{align*}
	\delta_{\bal,\bG} = \mymax{1 \le i,j \le n}\sqrt{|\alpha_i\alpha_j|} \nor{G_i - G_j}{a},\qquad 
	\delta_\bal = \ab{1 - \sum_{j=1}^{n} \alpha_j}, \qquad 
	\delta_\gw = \nor{\gw}{a},
\end{align*}
with $\gw = G - \sum_{j=1}^{n} \alpha_j G_j$.

\subsection{Choice of closest $G_j$ for standard perturbation theory}\label{sub:closest_Gj}

When there are several $G_j$'s, standard perturbation theory approximates $\sol(G)$ by using only one of those $G_j$'s for $j \in \{1,\dots,n\}$, so we need a way to choose it. Since there is \textit{a priori} no way to know which one is going to provide the best approximation, a possible simple choice is to take $G_j$ with
\begin{align}\label{eq:choice_j_normal}
j = \argmin{1 \le i \le n} \nor{G - G_i}{e}.
\end{align}
This choice leads to singularities in plots, because the index $j$ can change non-smoothly as we change the parameters $\bal$ in the definition of $G$. The standard perturbation theory approximation of order $\ell$ using $G_j$ is then 
\[
    \bbP^j_\ell := \sum_{p=0}^{\ell} \pt_p^j,
\]    
where $\pt_p^j$ is as in~\eqref{eq:dmpt} but with the index $1$ replaced by $j$. By fairness, as a choice of $j$ we took
\begin{align}\label{eq:choice_j_fair}
	j_\ell = \argmin{1 \le i \le n} \nor{\sol\pa{G_i} - \bbP^i_\ell}{e},
\end{align}
which is the best one that standard perturbation theory provides, for each $\ell$, and provides continuous density matrices as well as $\bbP^j_\ell$ as parameters change.

\subsection{Multipoint approximation when ${\delta_\gw = 0}$ and $\delta_{\bal,\bG} \rightarrow 0$}\label{sub:all_zero}

We study the behaviour of the multipoint approximation as $G \in \Aff \pa{G_j}_{j=1}^n$ and as $G_j \rightarrow G_1$ for all $j \in \{2,\dots,n\}$.
More precisely, we take $n=4$, $\bal = \pa{-0.3,0.4,0.3,0.6}$, and 
\begin{align*}
G_1 := V_1, \qquad G_j := V_1 + \ep V_j,\qquad  j \in \{2,3,4\}
\end{align*}
 and $G := \sum_{j=1}^{n} \alpha_j G_j$. Then, as $\ep \rightarrow 0$, we have $\delta_{\bal,\bG} = c \ep$ for some $c > 0$ independent of~$\ep$. In Figure~\ref{fig:comp_dmpt}, we display the error $D_e\pa{\sol_{\text{exact}} , \sol_{\text{approx}}}$ between the exact density matrix 
\begin{align*}
\sol_{\text{exact}} := \sol\pa{\sum_{j=1}^{n} \alpha_j G_j}
\end{align*}
 and several approximating quantities $\sol_{\text{approx}}$, i.e., the zeroth and first order of multipoint perturbation, as well as the zeroth, first, second and third orders of standard perturbation theory given in~\eqref{eq:dmpt} for comparison. The numerical results are displayed on Figure~\ref{fig:comp_dmpt}, where we simulate two different cases, namely $\delta_\bal = 0$ on the left and $\delta_\bal > 0$ on the right.

\begin{figure}[h]
\begin{center}
\includegraphics[width=5cm,trim={0cm 0cm 0cm 0cm},clip]{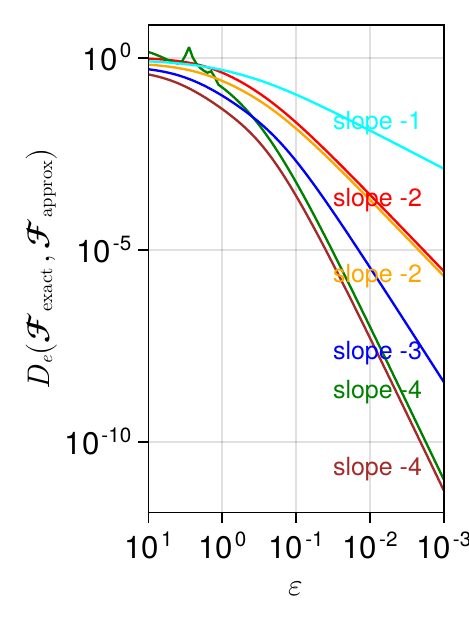} 
\includegraphics[width=5cm,trim={0cm 0cm 0cm 0cm},clip]{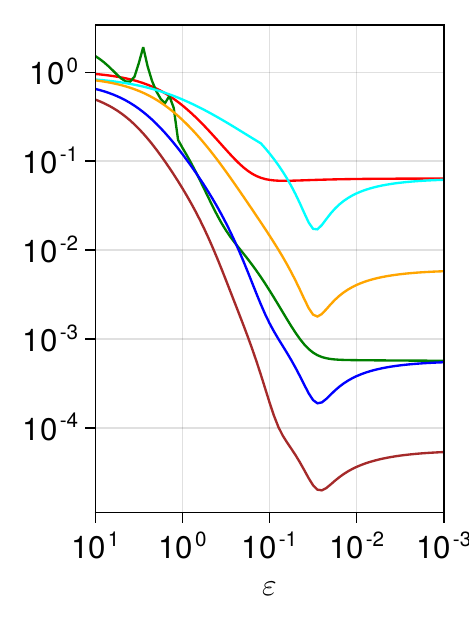} 
\includegraphics[width=3cm,trim={2.5cm 0cm 0cm 0cm},clip]{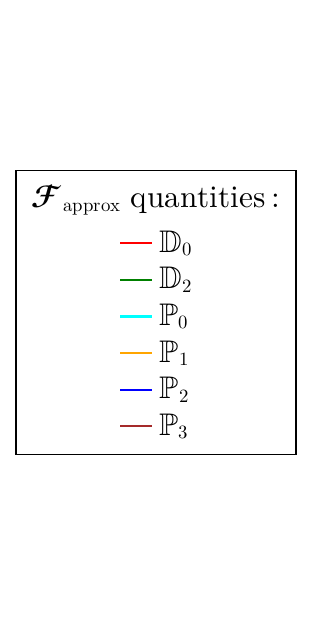} 
\caption{Comparision of the exact solution $\sol_{\text{exact}}$ with the approximations $\sol_{\text{approx}}$ based on multipoint perturbation theory~\eqref{eq:Dabc} and based on standard perturbation theory~\eqref{eq:dmpt}. (Left) $\bal = (-0.3,0.4,0.3,0.6)$ so that ${s_\bal = 1}$. The slopes of the asymptotic curves are indicated with the same color as the plot. (Right) $\bal = (-0.4,0.5,0.4,0.7)$ so ${s_\bal = 1.2 \neq 1}$. 
}\label{fig:comp_dmpt}
\end{center}
\end{figure}

As expected by the theoretical results presented in Section~\ref{sub:case_efficient}, we observe that the orders of convergence for the perturbative expansion are the expected ones, in particular the linear approximation of $\sol$ (i.e. $\D_0$) which corresponds to the zeroth expansion order in the multipoint perturbation theory is asymptotically of same accuracy as the first expansion order in standard perturbation theory when $\delta_{\bal,\bG} \rightarrow 0$. Comparing the slopes of the error plots, we observe that the bound~\eqref{eq:main_bound} is sharp in the case $\delta_\bal=0$.


\subsection{Multipoint approximation when $\delta_\gw = 0$, $\delta_\bal \neq 0$ and $\pa{G_j}_{j=1}^n$ constant}
\label{sub:delta_alpha}

We now consider $G$ as a linear combination of the $G_j$'s but not necessarily an affine one, i.e. the sum of $\alpha_j$'s can be different from one. For this purpose,  take 
\begin{align*}
	G_1 := V_1, \qquad G_2 := V_1 + \f{V_2}{5}. 
\end{align*}
In this situation, we show in Figure~\ref{fig:delta_alpha} the errors between the exact density matrix $\sol(\alpha_1 G_1 + \alpha_2 G_2)$ and approximation thereof using standard perturbation theory and multipoint perturbation theory while $\bal = (\alpha_1,\alpha_2)$ changes. 
As expected, we mainly observe that multipoint perturbation is more efficient than standard perturbation on the neighborhood of $\Aff \pa{G_1,G_2}$ which is indicated by the dotted line and corresponds to the case $\alpha_1+\alpha_2=1$, when we compare approximations of similar orders. More precisely, the error with $\D_0$ as the approximation is smaller than the error with $\bbP_0$ and the error with $\D_2$ is smaller than the error with $\bbP_2$ around $\Aff \pa{G_1,G_2}$.

\begin{figure}[h]
\hspace*{-0cm}
\begin{tikzpicture}[scale=1]
\node at (-6,0) {$\bbP_0$};
\node at (-3,0) {$\bbP_1$};
\node at (-0.3,0) {$\bbP_2$};
\node at (2.5,0) {$\D_0$};
\node at (5.5,0) {$\D_1$};
\node at (8,0) {$\D_2$};
\end{tikzpicture}
\hspace*{-1cm}
\includegraphics[width=18cm,trim={0cm 1.5cm 0cm 0.3cm},clip]{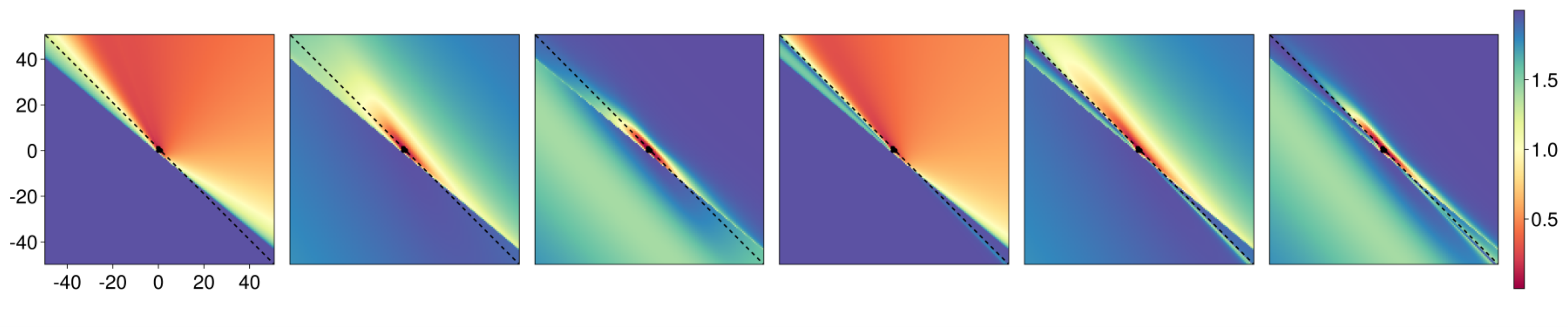} 
\hspace*{-1cm}
\includegraphics[width=18cm,trim={0cm 1.5cm 0cm 0.3cm},clip]{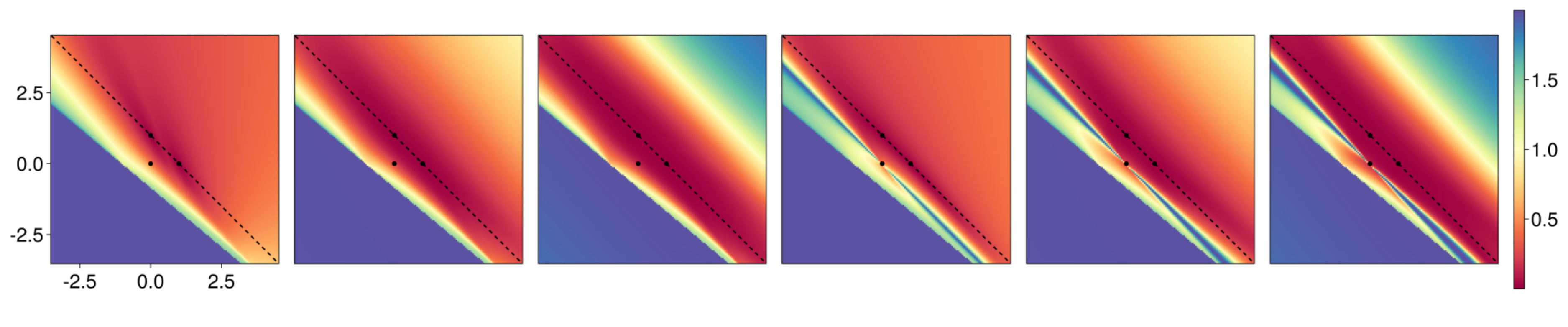} 
\hspace*{-1cm}
\includegraphics[width=18cm,trim={0cm 1.5cm 0cm 0.3cm},clip]{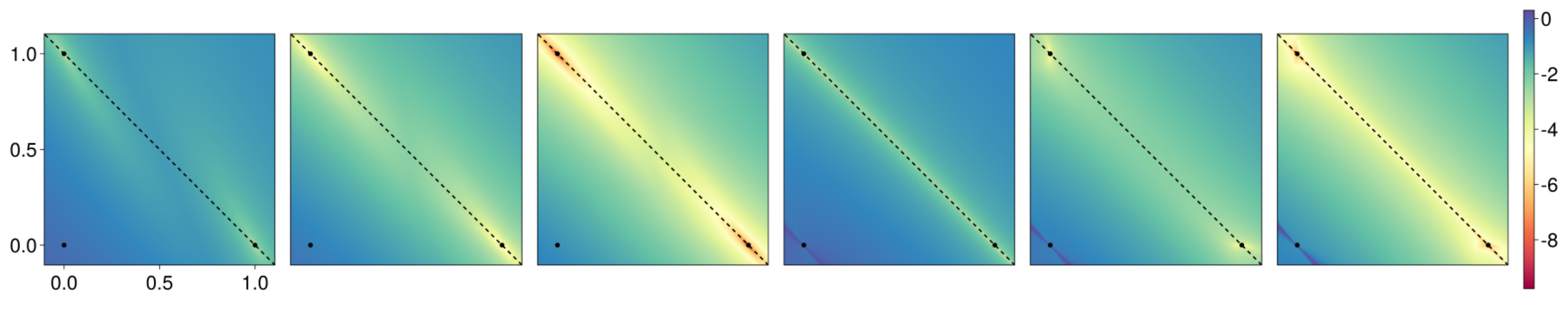}\\
\hspace*{-1cm}
\includegraphics[width=18cm,trim={0cm 1.5cm 0cm 0.3cm},clip]{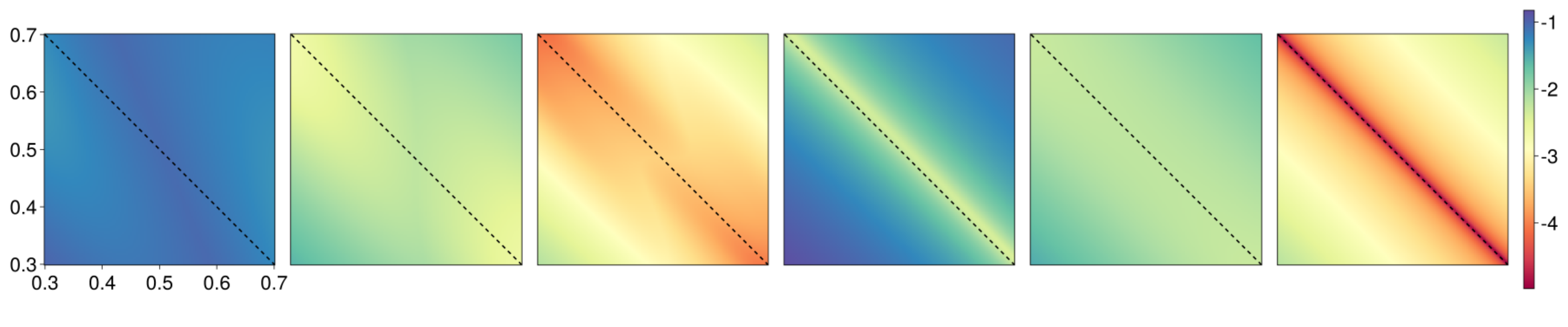} 
\hspace*{-1cm}
\includegraphics[width=18cm,trim={0cm 1.5cm 0cm 0.3cm},clip]{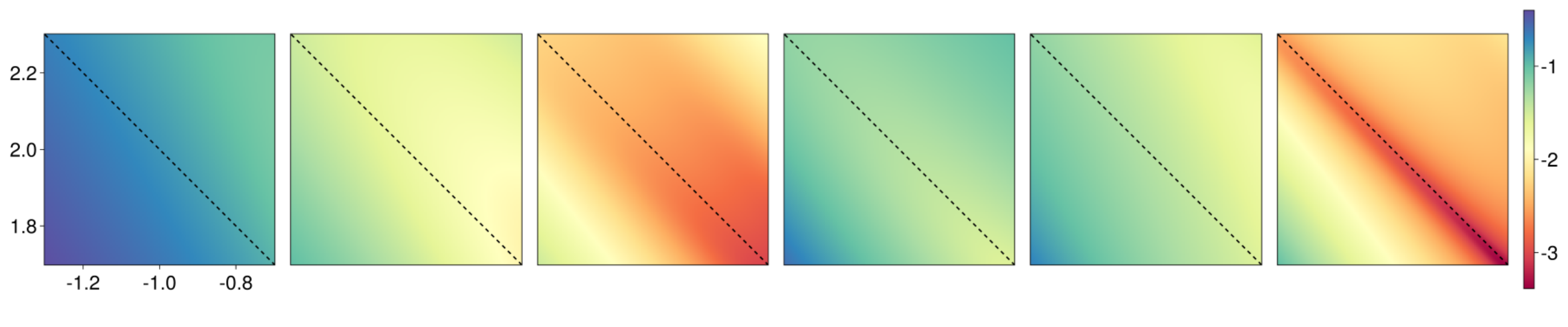}~
\caption{Error ${E(\sol_{\text{approx}}) := D_e\pa{\sol_{\text{exact}}, \sol_{\text{approx}}}}$, the $x$-axis is $\alpha_1$ and the $y$-axis is $\alpha_2$, forming $\bal = \pa{\alpha_1,\alpha_2}$. From left to right, the quantities $\sol_{\text{approx}}$ are respectively $\bbP_0$, $\bbP_1$, $\bbP_2$, $\D_0$, $\D_1$ and $\D_2$ as indicated on top. From top to bottom, we vary the range of $\alpha_1$ and $\alpha_2$, which is indicated on the left figure for each row. The two first rows indicate $E(\sol_{\text{approx}})$ while the three last ones draw $\log_{10} E(\sol_{\text{approx}})$. On the right of each row, we have the scale giving the correspondency between colors and values. We plot the particular points $\bal$ equal to $(0,0)$, $(0,1)$ and $(1,0)$, and we plot the line $\alpha_1 + \alpha_2 = 1$ with a dotted line.
}\label{fig:delta_alpha}
\end{figure}

\subsection{Multipoint approximation when $\delta_\gw = 0$, $\delta_\bal \rightarrow 0$ and $\pa{G_j}_{j=1}^n$ constant.}

We take the exact same situation as previously in Section~\ref{sub:delta_alpha}, that is
\begin{align*}
	G_1 := V_1, \qquad G_2 := V_1 + \f{V_2}{5}. 
\end{align*}
and study the limit $\delta_\bal \rightarrow 0$. We take $\bal := \pa{\f 12 + \ep, \f 12 + \ep}$ and display in Figure~\ref{fig:delta_alpha_dist} the relative errors $D_e\pa{\sol_{\text{exact}}, \sol_{\text{approx}}}$ against $\ep \rightarrow 0$. 
This corresponds to a zoom of Figure~\ref{fig:delta_alpha} around the affine space at $(0.5,0.5)$ along the diagonal $\alpha_1=\alpha_2$.

Further, the plateau signifies the regime where the error in $\delta_{\bal,\bG}$ (i.e. $\delta_{\bal,\bG}^{\ell +1 + \xi_\ell}$) dominates the one introduced by $\delta_\bal$ (i.e. $\delta_{\bal}^{\ell +1}$).
We observe that when $\ep$ is small enough (about $10^{-2}$), the error for the approximations based on multipoint perturbation are about one order of magnitude smaller than the errors for the approximations of similar approximation order using the standard perturbation method.

\begin{figure}[h]
\begin{center}
\includegraphics[width=13cm,trim={0cm 0cm 0cm 0cm},clip]{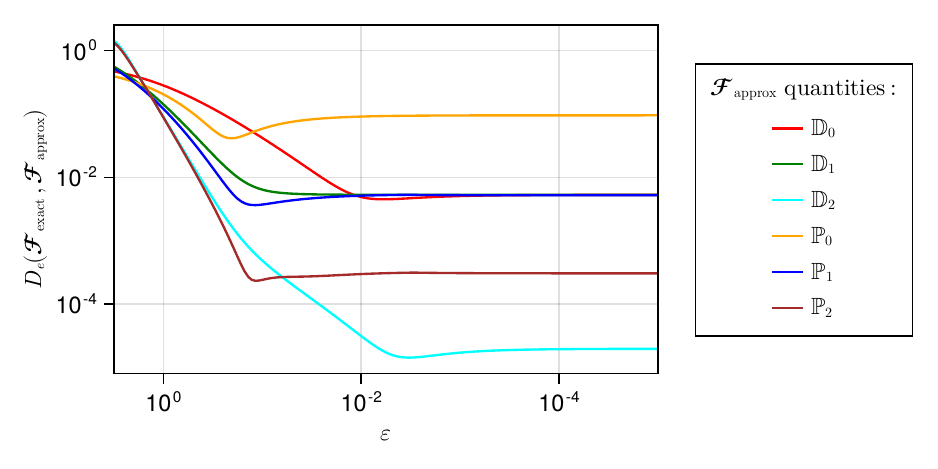} 
\caption{Relative errors $D_e\pa{\sol_{\text{exact}}, \sol_{\text{approx}}}$ with $\sol_{\text{approx}}$ taken as $\D_0,\D_1,\D_2,\bbP_0,\bbP_1,\bbP_2$ as $\delta_\bal \rightarrow 0$, when $\delta_\gw = 0$ and $(G_1,G_2)$ is constant. We parametrize $\bal = \pa{\f 12 + \ep, \f 12 + \ep}$.}\label{fig:delta_alpha_dist}
\end{center}
\end{figure}

\subsection{Multipoint approximation when $\delta_\gw \rightarrow 0$, with $\pa{G_j}_{j=1}^n$ constant}
\label{ssec:MP_deltag_small}

Let us take $n=2$,
\begin{align*}
G_1 := V_1, \qquad G_2 := V_1 + \f{V_2}{5}, \qquad \widetilde{G}_3 := V_1 + \f{V_3}{5},
\end{align*}
and for $(\beta_1,\beta_2) \in \R^2$, define $G := \beta_1 G_1 + \beta_2 G_2 + \ep \widetilde{G}_3$, where $\ep$ is going to converge to zero. Originally, one only knows $G$ and $G_j$, one does not know how $G$ was built, so we need to choose a way of finding $\bal = \pa{\alpha_j}_{j=1}^2$ to then build $\gw := G - \sum_{j=1}^{2} \alpha_j G_j$ and use the multipoint perturbation formula via the expansion in the $\D_\ell$'s. We consider three different ways of doing so.
First, we consider the minimization problem
\begin{align}\label{eq:opt_alpha}
	\mymin{\bal \in \R^2} \pa{D_a\pa{G,  \begingroup\textstyle\sum\endgroup_{j=1}^{2}\alpha_j G_j}^2 + \xi \ab{1 - \begingroup\textstyle\sum\endgroup_{j=1}^{2} \alpha_j}^2   }, 
\end{align}
and our three different ways of finding $\bal$ will correspond to the optimizers of~\eqref{eq:opt_alpha} with $\xi \in \{0,1,+\infty\}$. With an abuse of notation, the case $\xi = +\infty$ will refer to the problem
\begin{align*}
\mymin{\bal \in \R^2 \\ \sum_{j=1}^{n} \alpha_j  = 1} D_a\pa{G,\begingroup\textstyle\sum\endgroup_{j=1}^{2} \alpha_j G_j}^2.
\end{align*}
In Figure~\ref{fig:min_methods}, we plot the relative error  quantities $D_e\pa{\sol_{\text{exact}}, \sol_{\text{approx}}}$ against $\ep$, resulting in $\delta_\gw \rightarrow 0$, for the three values of $\xi$, and for two values of $(\beta_1,\beta_2)$, one for which $\beta_1 + \beta_2 \neq 1$ and the other one respecting $\beta_1 + \beta_2 = 1$. It seems that in any case, the best choice for $\xi$ is $0$.

\begin{figure}[h]
\begin{center}
\includegraphics[width=5cm,trim={0cm 0cm 0cm 0cm},clip]{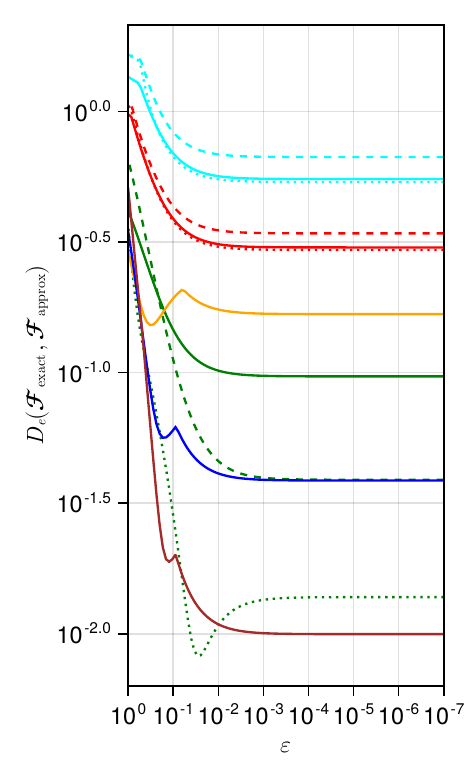} 
\includegraphics[width=5cm,trim={0cm 0cm 0cm 0cm},clip]{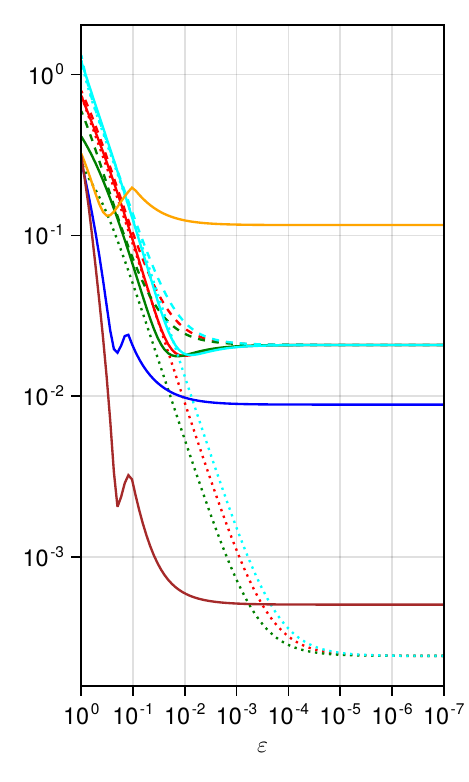} 
\includegraphics[width=5cm,trim={2cm 0cm 0cm 0cm},clip]{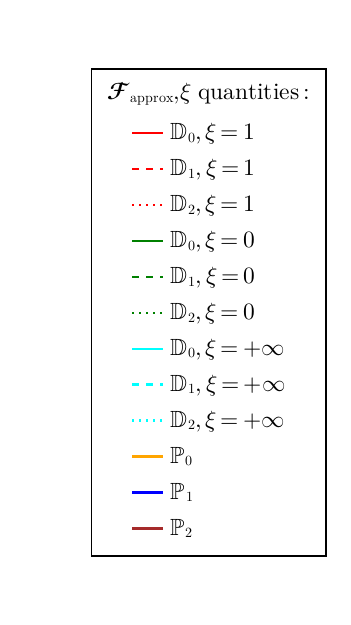} 
\caption{
Relative errors against the parameter $\ep$. We study here the convergence ${\delta_\gw \rightarrow 0}$, with $\beta$ and $\pa{G_j}_{j=1}^2$ constant, and three ways of finding~$\bal$. (Left) $(\beta_1,\beta_2) = \pa{0.3,1.2}$. (Right) $(\beta_1,\beta_2) = \pa{0.3, 0.7}$.
}\label{fig:min_methods}
\end{center}
\end{figure}



\subsection{First multipoint approximation, and then standard perturbation theory}%
\label{sub:Multipoint approximation and perturbation theory}

Since multipoint perturbation is particularly efficient on $\Aff \pa{G_j}_{j=1}^n$, it is natural in the case when $G \notin \Aff \pa{G_j}_{j=1}^n$ to try to first solve~\eqref{eq:opt_alpha}, then compute $\D_2$, and finally compute the second order approximation using standard perturbation theory between $\sum_{j=1}^{n} \alpha_j G_j$ and $G$. We observe that the resulting total error is very similar the last column in Figure~\ref{fig:delta_alpha}, and thus omit to report the results here.

\section{Proofs}%
\label{sec:Proofs}
In this section we provide the proofs of the propositions.

\subsection{Bounds}
\label{subsec:Preliminary bound}

Let us first start by showing a standard result, enabling to have a bound on $\pa{z - H(G)}^{-1}$ explicit in $z$, and relating the bound to the definition~\eqref{eq:set_operators} of $\cG$.

\begin{lemma}[Explicit bound on the resolvent]\label{lem:bound_res}
	Take an essentially self-adjoint and bounded from below operator $H^0$ on the Hilbert space $\cH$. Take $G \in \cG$, $\mu,\eta > 1$ with $\eta$  large enough so that 
\begin{align}\label{eq:assum_G}
\norm{ \pa{H^0 + \eta}^{-\f 12} G \pa{H^0 + \eta}^{-\f 12}} \le \f 12.
\end{align}
The last condition is satisfied for instance for $\eta = \max(4 c^G_{1/8},1)$, where $c^G_{\ep}$ is defined in~\eqref{eq:set_operators}. Take a contour $\cC \subset \C$ such that $\d \pa{\cC, \sigma(H(G))} =: \xi > 0$. Then there holds
\begin{multline}\label{eq:ineq_f}
	\mymax{z \in \cC} \nor{\pa{z - H(G)}^{-1} }{e} \\
	\le 2\pa{1 + \f{\ab{\mu - \eta} }{\eta + \min \sigma(H(G))}}\pa{1+ 2\ab{\eta + \mymax{z \in \cC} \ab{z}} \pa{1+ \f{\ab{\eta + \mymax{z \in \cC} \ab{z}}}{\xi}}}.
\end{multline}
\end{lemma}
We used the notation $\d \pa{S_1,S_2} := \min_{s_1 \in S_1, s_2 \in S_2} \ab{s_1-s_2}$, and the energy norm which depends on $\mu$ is defined in~\eqref{eq:norm_dmatrices}.
\begin{proof}[Proof of Lemma~\ref{lem:bound_res}]
	We denote by $\lambda := \min \sigma(H^0) \ge 0$. Also define
\begin{align*}
	D &:= \pa{H^0 + \eta}^{\f 12} \pa{H(G) + \eta}^{-1} \pa{H^0 + \eta}^{\f 12}, \\
	J &:= \pa{H^0 + \eta}^{-\f 12} G \pa{H^0 + \eta}^{-\f 12},
\end{align*}
and there holds $\norm{J} \le \f 12$ by assumption \eqref{eq:assum_G}. In the sense of forms, $1 \le \eta \le H^0 + \eta$, so $\norm{\pa{H^0 + \eta}^{-\f 12}} \le 1$. Moreover
\begin{align*}
D^{-1} = \pa{H^0 + \eta}^{-\f 12} \pa{H(G) + \eta} \pa{H^0 + \eta}^{-\f 12} = 1 + J
\end{align*}
and hence $\norm{D} \le \pa{1 - \norm{J}  }^{-1} \le 2$. Now 
\begin{align*}
	& \pa{H^0 + \eta}^{\f 12} \pa{H(G) - z}^{-1} \pa{H^0 + \eta}^{\f 12} \\
	& \qquad  = \pa{H^0 + \eta}^{\f 12} \pa{H(G) + \eta}^{-1} \pa{H(G) + \eta} \pa{H(G) -z}^{-1} \pa{H^0 + \eta}^{\f 12} \\
	& \qquad = D \pa{H^0 + \eta}^{-\f 12}  \pa{H(G) + \eta} \pa{H(G) -z}^{-1} \pa{H^0 + \eta}^{\f 12} \\
	& \qquad = D \pa{1 + \pa{\eta + z} \pa{H^0 + \eta}^{-\f 12} \pa{H(G) -z}^{-1} \pa{H^0 + \eta}^{\f 12}} \\
	& \qquad = D \pa{1 + \pa{\eta + z}\pa{H^0 + \eta}^{-\f 12} \pa{1+ \pa{\eta + z} \pa{H(G) -z}^{-1}} \pa{H^0 + \eta}^{-\f 12} D},
\end{align*}
and we deduce that
\begin{align*}
\norm{\pa{H^0 + \eta}^{\f 12} \pa{H(G) - z}^{-1} \pa{H^0 + \eta}^{\f 12}} \le 2 \pa{1 + 2 \ab{\eta + z} \pa{1+\f{\ab{\eta + z}}{\d \pa{z,\sigma\pa{H(G)}}} } }.
\end{align*}
Then
\begin{align*}
	\pa{H^0 + \mu} \pa{H^0 + \eta}^{-1} = 1 + \pa{\mu - \eta} \pa{H^0 + \eta}^{-1}
\end{align*}
so 
\begin{align*}
	\norm{\pa{H^0 + \mu} \pa{H^0 + \eta}^{-1}} \le 1 + \f{\ab{\mu - \eta} }{\lambda + \eta}.
\end{align*}
We have $H^0 + \mu \ge 1$ and $\pa{H^0 + \eta}^{-1} \ge 0$ in the sense of forms, so $\pa{H^0 + \mu} \pa{H^0 + \eta}^{-1} \ge 0$. By monotonicity of the square root of operators, we deduce that
\begin{align*}
	\norm{\pa{\pa{H^0 + \mu} \pa{H^0 + \eta}^{-1}}^{\f 12}} \le \pa{1 + \f{\ab{\mu - \eta} }{\lambda + \eta}}^{\f 12}.
\end{align*}
We deduce~\eqref{eq:ineq_f} by writting
\begin{multline*}
\pa{H^0 + \mu}^{\f 12} \pa{H(G) - z}^{-1} \pa{H^0 + \mu}^{\f 12}\\
= \sqrt{\pa{H^0 + \mu} \pa{H^0 + \eta}^{-1}} \pa{H^0 + \eta}^{\f 12} \pa{H(G) - z}^{-1} \pa{H^0 + \eta}^{\f 12} \sqrt{\pa{H^0 + \mu} \pa{H^0 + \eta}^{-1}}.
\end{multline*}
We deduce that
\begin{align*}
	\nor{\pa{z - H(G)}^{-1} }{e} \le 2\pa{1 + \f{\ab{\mu - \eta} }{\eta + \lambda}}\pa{1+ 2\ab{\eta + z} \pa{1+ \f{\ab{\eta + z}}{\d \pa{z,\sigma\pa{H(G)}}}}},
\end{align*}
from which we easily get~\eqref{eq:ineq_f}.

We then search for the more explicit condition $\eta =\max( 4 c^G_{1/8},1)$ that fulfills~\eqref{eq:assum_G}.
Take $G \in \cG$ and $\vp \in \cH$ such that the following expressions are finite, we have 
\begin{align*}
	&\norm{\ab{G}^{\f 12} \pa{H^0 + \eta}^{-\f 12} \vp}^2 = \ps{\pa{H^0 + \eta}^{-\f 12}\vp, \ab{G} \pa{H^0 + \eta}^{-\f 12} \vp} \\
							     &\qquad  \underset{\substack{~\eqref{eq:set_operators}}}{\le} \; \ps{\pa{H^0 + \eta}^{-\f 12}\vp, \pa{\ep H^0 + c^G_\ep} \pa{H^0 + \eta}^{-\f 12} \vp} \\
& \qquad = \ps{\vp, \pa{\ep + \f{c^G_\ep - \ep \eta}{H^0 + \eta}} \vp} \le \pa{\ep + \f{\ab{c^G_\ep - \ep \eta }}{\lambda + \eta}} \norm{\vp}^2  \\
& \qquad \le \pa{2\ep + \f{c^G_\ep}{\lambda + \eta}} \norm{\vp}^2\le \pa{2\ep + \f{c^G_\ep}{\eta}} \norm{\vp}^2,
\end{align*}
so by the polar decomposition $G = \ab{G} U$,
\begin{align*}
	& \norm{\pa{H^0 + \eta}^{-\f 12} G \pa{H^0 + \eta}^{-\f 12}} = \norm{\pa{H^0 + \eta}^{-\f 12} \ab{G}^{\f 12} U \ab{G}^{\f 12}\pa{H^0 + \eta}^{-\f 12}} \\
								   & \qquad \le \norm{\ab{G}^{\f 12}\pa{H^0 + \eta}^{-\f 12}}^2 \norm{U} \le 2\ep + \f{c^G_\ep}{\eta}.
\end{align*}
Choosing $\ep = 1/8$ and $\eta$ such that $\f{c^G_{1/8}}{\eta} \le 1/4$ guaranties that~\eqref{eq:assum_G} is satisfied.
\end{proof}

\subsection{Standard perturbation theory}%
\label{sub:Standard perturbation theory}

We then provide a proof of the standard perturbation result.

\begin{proof}[Proof of Proposition~\ref{prop:standard_pt}]
We remark that
\begin{align}
\nonumber
\lefteqn{\hzmp \pa{z-H(G)}^{-1}\hzmp}
\\ \nonumber
&\qquad = \hzmp \pa{z-H(G_1)}^{-1} \hzmp \\
&\qquad \qquad \times\hzmm\pa{1 - \gw \pa{z-H(G_1)}^{-1}}^{-1} \hzmp 
\nonumber \\ \nonumber
&\qquad = \hzmp \pa{z-H(G_1)}^{-1} \hzmp 
\\
& \qquad \qquad \times \pa{1 - \hzmm \gw \hzmm \hzmp\pa{z-H(G_1)}^{-1}\hzmp}^{-1}.
\label{eq:trick}
\end{align}
Provided that the condition $\nor{\gw}{a} \le 1/(2C)$ is satisfied,~\eqref{eq:pert_thy} is well-defined in the corresponding norms and we have
\begin{align*}
	\nor{\pa{z-H(G)}^{-1}}{e} \le \f{\nor{ \pa{z-H(G_1)}^{-1}}{e}}{1 - \nor{\gw}{a} \nor{ \pa{z-H(G_1)}^{-1}}{e}} \le 2 C.
\end{align*}
We can then make a Neumann expansion and have
\begin{align*}
	&\pa{z-H(G)}^{-1} - \pa{z-H(G_1)}^{-1}\sum_{p=0}^{\ell} \pa{\gw\pa{z-H(G_1)}^{-1}}^p \\
	&\qquad \qquad =  \pa{z-H(G_1)}^{-1}\sum_{p=\ell+1}^{+\infty} \pa{\gw\pa{z-H(G_1)}^{-1}}^p\\
	&\qquad\qquad  =  \pa{z-H(G_1)}^{-1} \pa{\gw\pa{z-H(G_1)}^{-1}}^{\ell+1} \pa{1 - \gw \pa{z-H(G_1)}^{-1}}^{-1}
\end{align*}
and we obtain
\begin{multline*}
	\nor{\pa{z-H(G)}^{-1} - \pa{z-H(G_1)}^{-1}\sum_{p=0}^{\ell} \pa{\gw\pa{z-H(G_1)}^{-1}}^p}{e}  \\
	 \le \nor{\gw}{a}^{\ell+1}  \f{\nor{ \pa{z-H(G_1)}^{-1}}{e}^{\ell +2}}{1 - \nor{\gw}{a} \nor{ \pa{z-H(G_1)}^{-1}}{e}} \le 2 C^{\ell +2} \nor{\gw}{a}^{\ell+1}.
\end{multline*}
Still in series of $\gw$, by integration over $\cC$ and using Cauchy's formula~\eqref{eq:cauchy_formula}, we can deduce the expansion 
\begin{align*}
\sol(G) = \sum_{m=0}^{+\infty} \pt_m.
\end{align*}
We integrate over $z \in \cC$ and obtain the result~\eqref{eq:bound_standard_pert}.
\end{proof}

\subsection{Main formula}%
\label{sub:Main formula}

We will use the notation $R_j := \pa{z-H(G_j)}^{-1}$, $G_{ij} := G_i - G_j$ so we have, by the resolvent formula,
\begin{align}\label{eq:diff_res}
\pa{z-H(G)}^{-1} - R_j = \pa{z-H(G)}^{-1} (G-G_j) R_j, \qquad\qquad R_j - R_i = R_j G_{ji} R_i.
\end{align}


\begin{proof}[Proof of Theorem \ref{thm:res_form}]
For any ${\bm \alpha} = (\alpha_i)_{i=1}^n\in\R^n,$ let $s_\bal := \sum_{j=1}^{n} \alpha_j$. There holds
\begin{align*}
    s_\bal\left(1 - \pa{z - H(G)} \bbL \right)
&= \pa{z-H(G)}\pa{ \pa{\sum_{j=1}^{n} \alpha_j} \pa{z-H(G)}^{-1} - \sum_{j=1}^{n} \alpha_j R_j} \\
& = \pa{z-H(G)}\sum_{j=1}^{n} \alpha_j \pa{\pa{z-H(G)}^{-1} - R_j}.
\end{align*}
Using~\eqref{eq:diff_res}, and recalling that $\gw := G - \sum_{j=1}^{n} \alpha_j G_j$, we obtain
\begin{align*}
     s_\bal\left(1 - \pa{z - H(G)} \bbL \right)
     &= \sum_{j=1}^{n}  \alpha_j \pa{G - G_j} R_j \\
&  =   \gw\sum_{j=1}^{n}  \alpha_j R_j + \sum_{j=1}^{n}  \alpha_j \pa{\pa{\sum_{i =1}^n  \alpha_i G_i} - G_j} R_j \\
&  = \gw\sum_{j=1}^{n}  \alpha_j R_j + \sum_{i,j=1}^{n} \alpha_j \alpha_i \pa{G_i - G_j} R_j + \pa{s_\bal - 1} \sum_{j=1}^{n}  \alpha_j  G_j R_j \\
& = s_\bal \gw \bbL + \sum_{i,j=1}^{n} \alpha_j \alpha_i \pa{G_i - G_j} R_j 
+ \pa{s_\bal - 1} s_\bal \bbA.
\end{align*}
Noting that $G_{ji} = - G_{ij}$, we rearrange the second term on the right-hand side into
\begin{align*}
\sum_{\substack{1 \le i,j \le n}}  \alpha_i\alpha_j  G_{ij} R_j  
&=  \f 12 \sum_{\substack{1 \le i,j \le n}}  \alpha_i\alpha_j  G_{ij} R_j + \f 12 \sum_{\substack{1 \le i,j \le n}}  \alpha_j \alpha_i  G_{ji} R_i \\
& =  \f 12 \sum_{\substack{1 \le i , j \le n}} \alpha_i\alpha_j  G_{ij} \pa{ R_j - R_i}. 
\end{align*}
Using again~\eqref{eq:diff_res}, we get
\begin{align*}
\sum_{\substack{1 \le i,j \le n}}  \alpha_i\alpha_j  G_{ij} R_j 
&=  -\f 12 \sum_{\substack{1 \le i , j \le n}} \alpha_i\alpha_j  G_{ij}  R_j G_{ij}R_i \\
&=  - \sum_{\substack{1 \le i < j \le n}} \alpha_i\alpha_j  G_{ij}  R_j G_{ij}R_i \\
&= - s_\bal \bbH.
\end{align*}
Therefore,
\[
\left(1 - \pa{z - H(G)} \bbL \right) = \gw\bbL - \bbH + (s_\bal-1) \bbA,
\]
which is equivalent to 
\begin{align*}
 \pa{z-H(G)}^{-1} (1 + \bbH +  \Del \bbA - \gw \bbL) = \bbL.
\end{align*}
\end{proof}

\subsection{Multipoint perturbation theory bound}



\begin{proof}[Proof of Corollary~\ref{cor:multipoint_bound}]
We have
\begin{align*}
    \hzmm \bbH \hzmp &= s_\bal^{-1} \sum_{\substack{1 \le i < j \le n}} \alpha_i \alpha_j \hzmm G_{ij} R_j G_{ij} R_i \hzmp
    \\
    &= s_\bal^{-1} \sum_{\substack{1 \le i < j \le n}} \alpha_i \alpha_j \hzmm G_{ij}\hzmm \\
    &\quad \times \hzmp R_j \hzmp \hzmm G_{ij} \hzmm \\
    & \quad \times \hzmp R_i \hzmp.
\end{align*}
Recalling that $\delta_{\bal,\bG} := \mymax{1 \le i,j \le n} \sqrt{|\alpha_i\alpha_j|} \nor{G_{ij}}{a}$, we obtain
\begin{align*}
\norm{\hzmm \bbH \hzmp} \le \delta_{\bal,\bG}^2 \tfrac {n(n-1)}2 s_\bal^{-1}  \mymax{1 \le j \le n \\ z \in \cC } \nor{R_j}{e}^2.
\end{align*}
Similarly,
\begin{align*}
\norm{\hzmm \bbA \hzmp} &\le n s_\bal^{-1} \mymax{1 \le j \le n \\ z \in \cC } \ab{\alpha_j} \nor{G_j}{a} \nor{R_j}{e} \\
\norm{\hzmp \bbL \hzmp} &\le n s_\bal^{-1} \mymax{1 \le j \le n \\ z \in \cC } \ab{\alpha_j} \nor{R_j}{e}.
\end{align*}
Standard arguments, expressed in form of Lemma~\ref{lem:bound_res}, enable us to bound $\nor{R_j}{e}$ using the definition of $\cG$. Thus we have a more explicit bound on 
\begin{align}\label{eq:beta}
	\gamma &:= n s_\bal^{-1}\mymax{1 \le j \le n \\ z \in \cC }  \pa{\nor{R_j}{e} \seg{ \ab{\alpha_j}\pa{1+\nor{G_j}{a} } + \tfrac 12  (n-1)  \nor{R_j}{e} }} \le \beta c_\bal,
\end{align} 
where $\beta$ depends neither on $G$ nor on $\bal$, and we obtain
\begin{align*}
\norm{\hzmm \bbH \hzmp} &\le \gamma \delta_{\bal,\bG}^2, \\
\norm{\hzmm \Del \bbA \hzmp} &\le \gamma \delta_\bal, \\
\norm{\hzmm  \gw \bbL \hzmp} &\le \gamma \delta_\gw.
\end{align*}
Hence,
\begin{align*}
\norm{\hzmm  \pa{\bbH - \Delm \bbA - \gw \bbL} \hzmp} \le \gamma \pa{\delta_{\bal,\bG}^2 + \delta_\bal + \delta_\gw},
\end{align*}
and assuming that $\delta_{\bal,\bG}^2 + \delta_\bal + \delta_\gw < (2\gamma)^{-1}$, we can proceed to the Neumann expansion of $\pa{1 + \bbH - \Delm \bbA - \gw \bbL}^{-1}$, and have
\begin{align*}
\pa{1 + \bbH - \Delm \bbA - \gw \bbL}^{-1} = \sum_{p=0}^{+\infty}  \pa{-\bbH + \Delm \bbA + \gw \bbL}^p.
\end{align*}
Now, to obtain the corresponding series for $\sol(G)$, we need to compute
\begin{align*}
\sol(G) &=  \f{1}{2\pi i} \oint_\cC\pa{z - H\pa{G}}^{-1} \d z \\
&= \f{1}{2\pi i} \oint_\cC\bbL \pa{1 + \bbH - \Delm \bbA - \gw \bbL}^{-1} \d z \\
&= \sum_{p=0}^{+\infty} \f{1}{2\pi i} \oint_\cC \bbL \pa{-\bbH + \Delm \bbA + \gw \bbL}^p \d z  \\
&= \sum_{p=0}^{+\infty} \widetilde{\cD}_p,
\end{align*}
where
\begin{align*}
	\widetilde{\cD}_{(a,b,c)} &:= \cD_{(2a,b,c)}, \qquad \qquad \widetilde{\cD}_{p} := \sum_{\substack{a,b,c \in \N \\ a+b+c = p}} \widetilde{\cD}_{(a,b,c)}, 
	\qquad \qquad
		\widetilde{\D}_{\ell} := \sum_{p=0}^\ell \widetilde{\cD}_{p}.
\end{align*}
\begin{remark} We give here a bound which is not used in the following, but that we find interesting to present
\begin{align*}
	\sol(G) - \widetilde{\D}_\ell = \sum_{p=\ell +1}^{+\infty}  \f{1}{2\pi i} \oint_\cC \bbL \pa{-\bbH + \Delm \bbA + \gw \bbL}^p \d z 
\end{align*}
so
\begin{multline*}
\hzmp \pa{\sol(G) - \widetilde{\D}_\ell} \hzmp = \sum_{p=\ell +1}^{+\infty}  \f{1}{2\pi i} \oint_\cC \hzmp \bbL \hzmp \\
\times \pa{\hzmm \pa{-\bbH + \Delm \bbA + \gw \bbL} \hzmp}^p \d z 
\end{multline*}
and we can conclude that
\begin{align*}
\nor{\sol(G) - \widetilde{\D}_\ell}{e} \le c \pa{\gamma \pa{\delta_{\bal,\bG}^2 + \delta_\bal + \delta_\gw}}^{\ell +1} = c \gamma^{\ell +1}\sum_{\substack{a,b,c \in \N \\ a+b+c = \ell +1}} \delta_{\bal,\bG}^{2a} \delta_\bal^b \delta_\gw^c.
\end{align*}
\end{remark}
Now summing differently, and since $\cD_{\pa{2a+1,b,c}} = 0$ for any $a,b,c \in \N$,
\begin{align*}
\sol(G) = \sum_{p=0}^{+\infty} \widetilde{\cD}_p = \sum_{p=0}^{+\infty} \cD_p,
\end{align*}
so
\begin{align*}
	\sol(G) - \D_\ell = \sum_{p=\ell +1}^{+\infty} \cD_p= \sum_{\substack{a,b,c \in \N \\ a + b + c \ge \ell +1}} \cD_{\pa{a,b,c}},
\end{align*}
but
\begin{align*}
\nor{\widetilde{\cD}_{(a,b,c)}}{e} \le c \delta_{\bal,\bG}^{2a} \delta_\bal^b \delta_\gw^c, \qquad \qquad  \nor{\cD_{(a,b,c)}}{e} \le c \delta_{a \in 2\N}\delta_{\bal,\bG}^{a} \delta_\bal^b \delta_\gw^c,
\end{align*}
where $c \le C' c_\bal^m$ for some $m \in \N$ and $C'$ which does not depend on $G$ or $\bal$. We conclude by summation, using
\begin{align*}
	\nor{\sol(G) - \D_\ell}{e} \le C'_\ell c_\bal^{m_\ell} \sum_{\substack{a,b,c \in \N \\ a+b+c \ge \ell +1}} \delta_{a \in 2\N}\delta_{\bal,\bG}^{a} \delta_\bal^b \delta_\gw^c  \le C'_\ell c_\bal^{m_\ell} \sum_{\substack{a \in 2\N   \\ b,c \in \N \\ a+b+c \in \{ \ell +1 , \ell +2\}}} \delta_{\bal,\bG}^{a + 1} \delta_\bal^b \delta_\gw^c
\end{align*}
as an intermediate step.

Let us now treat the case where $1/2 > (2\gamma)^{-1}$ and 
\begin{align*}
 (2\gamma)^{-1} \le \delta_{\bal,\bG}^2 + \delta_\bal + \delta_\gw \le 1/2,
\end{align*}
so we also have that
\begin{align*}
 q_0 \le \delta_{\bal,\bG}^{\ell +1 + \xi_\ell} + \delta_\bal^{\ell +1} + \delta_\gw^{\ell+1},
\end{align*}
for some $q_0>0$ independent of $\bal$ and $\gw$. The bound remains true, but the preconstant $C_\ell$ may need to be increased, we detail here how we can provide a very coarse one. In this case $\delta_\gw \le 1/2$ and $\delta_\bal \le 1/2$, $c_\alpha$ is bounded from below uniformly in $\bal$, i.e. $c_\alpha \ge q_1>0$, and for instance we can take $m_\ell =1$ and
\begin{align*}
     \mymax{G \in \cG \\ \delta_\gw \le 1/2} \f{\nor{\sol(G) - \D_\ell}{e}}{c_\bal \pa{\delta_{\bal,\bG}^{\ell +1 + \xi_\ell} + \delta_\bal^{\ell +1} + \delta_\gw^{\ell+1}}} \le \f{1}{q_0 q_1} \mymax{G \in \cG \\ \delta_\gw \le 1/2 } \nor{\sol(G) - \D_\ell}{e} =: C_\ell, 
\end{align*}
this quantity being finite. This concludes the proof.
\end{proof}

\subsection{Proof of Proposition~\ref{prop:integrals}}
\begin{proof}[Proof of Proposition~\ref{prop:integrals}]
For any $a \in \{1,\dots,n\}$, any $z \in \pa{ \C \backslash \sigma(H(G_a))} \cup \{\lambda^k(G_a)\}$, we define
\begin{align*}
	R^\perp_a(z) := 
\left\{
\begin{array}{ll}
\pa{\pa{z - H(G_a)}_{\mkern 1mu \vrule height 2ex\mkern2mu \pa{\Ker \pa{\lambda^k(G_a) - H(G_a)}}^\perp}}^{-1} & \mbox{on }  \pa{\Ker \pa{\lambda^k(G_a) - H(G_a)}}^\perp, \\
0 & \mbox{on } \Ker \pa{\lambda^k(G_a) - H(G_a)},
\end{array}
\right.
\end{align*}
extended on all of $\cH$ by linearity. As a function of $z$, it is holomorphic in a neighborhood of $\{\lambda^k(G_a)\}$ containing $\lambda^k(G_a)$, and note that $R_a^\perp(\lambda^k(G_j)) = K_{ja}$. Take $a \in \{1,\dots,n\}$, we define $P_a^\perp = 1-P_a$
so for any $z \in \C \backslash \sigma(H(G_a))$,
	\begin{align*}
	\pa{z-H(G_a)}^{-1} = \pa{P_a + P_a^\perp} \pa{z-H(G_a)}^{-1} \pa{P_a + P_a^\perp} = \pa{z - \lambda^k(G_a)}^{-1} P_a + R_a^\perp(z)
	\end{align*}
	so
	\begin{multline*}
		\pa{z-H(G_a)}^{-1} A \pa{z-H(G_b)}^{-1} \\
		= R_a^\perp(z) A R_b^\perp(z) + \pa{z - \lambda^k(G_b)}^{-1} R_a^\perp(z) A P_b + \pa{z - \lambda^k(G_a)}^{-1} P_a A R_b^\perp(z) \\
		+ \pa{z - \lambda^k(G_a)}^{-1}\pa{z - \lambda^k(G_b)}^{-1} P_a A P_b.
	\end{multline*}
We recall Cauchy's residue formula, 
\begin{align*}
f^{(n)}(w) = \f{n!}{2\pi i}\oint_{\cC_0} \f{f(s) \d s}{(s-w)^{n+1}} 
\end{align*}
holding for any holomorphic function $f$, any $n \in \N$, any $w \in \C$ and any contour $\cC_0 \subset \C$ containing $w$. Let us assume that ${\lambda^k(G_a) \neq \lambda^k(G_b)}$. We see that the first term $R_a^\perp(z) A R_b^\perp(z)$ is holomorphic in $z$, hence the integral is zero. Then
\begin{multline*}
\f{1}{2\pi i}\oint_{\cC}\pa{z - \lambda^k(G_a)}^{-1}\pa{z - \lambda^k(G_b)}^{-1} \d z \\
= \pa{\lambda^k(G_b) - \lambda^k(G_a)}^{-1} + \pa{\lambda^k(G_a) - \lambda^k(G_b) }^{-1} = 0,
\end{multline*}
where we decomposed the integral $\oint_{\cC}$ into two integrals, one around the singularity associated to $\lambda^k(G_a)$, and the other one around the singularity associated to $\lambda^k(G_b)$. So the fourth term also gives a vanishing contribution. Finally,
\begin{align*}
\f{1}{2\pi i}\oint_{\cC} \pa{z - \lambda^k(G_b)}^{-1} R_a^\perp(z) A P_b \d z = R_a^\perp(\lambda^k(G_b)) A P_b = K_{ba} A P_b,
\end{align*}
and we use a similar computation for the remaining integral, and we deduce $\cI_{a,b}(A)$ in~\eqref{eq:comp_cI}. To compute $\cI_{a,b}(A)$ when $\lambda^k(G_a) = \lambda^k(G_b)$, we use the previous one and since the result is regular as $\ab{\lambda^k(G_a) - \lambda^k(G_b)}$ is arbitrarily small, we can deduce that the same formula holds in any case. We compute $\cI_{a,b,c}(A,B)$ with a similar procedure. We nevertheless give detail about the computation of terms of the following kind,
\begin{multline*}
\f{1}{2\pi i}\oint_{\cC}\pa{z - \lambda^k(G_a)}^{-1}\pa{z - \lambda^k(G_b)}^{-1} P_a A P_b B R_c^\perp(z) \d z \\
= \pa{\lambda^k(G_a) - \lambda^k(G_b)}^{-1} P_a A P_b B R_c^\perp(\lambda^k(G_a)) + \pa{\lambda^k(G_b) - \lambda^k(G_a) }^{-1} P_a A P_b B R_c^\perp(\lambda^k(G_b)) \\
= \pa{\lambda^k(G_a) - \lambda^k(G_b)}^{-1} P_a A P_b B \pa{K_{ac} - K_{bc}}.
\end{multline*}
Finally, by the resolvent formula, $K_{ac} - K_{bc} = \pa{\lambda^k(G_b) - \lambda^k(G_a)} K_{ac} K_{bc}$.
\end{proof}

\section{Conclusion}
\label{sub:Conclusion}
In this article, we introduced a multipoint perturbation formula for eigenvalue computations. 
It allows one to use the density matrix for several Hamiltonians $H(G_j)=H^0+G_j$ simultaneously to obtain the solution for a nearby Hamiltonian $H(G)=H^0+G$.
Our formula is based on the resolvent formalism and incorporates a new resolvent identity~\eqref{eq:main_formula} involving several Hamiltonians. 
Based on this identity, we then derived approximations of different expansion orders with respect to smallness parameters.
We also derived a detailed complexity analysis allowing one to compare the multipoint perturbation method to the standard perturbation method in terms of convergence order and complexity with the purpose to understand in which regimes multipoint perturbation is more efficient than standard perturbation theory.

We verified the asymptotic estimates~\eqref{eq:main_bound} by a series of numerical results for the discretized Schr\"odinger equation.
We observed, as expected by the theory, that multipoint perturbation is more efficient when the new $G$ is close to the affine space
\begin{align*}
    \Aff \pa{G_j}_{j=1}^{n} = \acs{\sum_{j=1}^{n} \alpha_j G_j \;\middle|\; \alpha_j \in \R, \sum_{j=1}^{n} \alpha_j = 1},
\end{align*}
and when the $G_j$'s are sufficiently close to each other. In such a case, the multipoint perturbation method is one order more accurate.




\section*{Research data management}

The code enabling to produce the figures of this document can be found on the Github repository \url{https://github.com/lgarrigue/multipoint_perturbation} 
commit \linebreak ee63c611d7135865b1743f257db6d2dec0d4f931 and on Zenodo \url{https://doi.org/10.5281/zenodo.7929850}

\section*{Acknowledgements}
The authors L.G. and B.S. acknowledge funding by the Deutsche Forschungsgemeinschaft (DFG, German Research Foundation) - Project number 442047500 through the Collaborative Research Center ``Sparsity and Singular Structures'' (SFB 1481).

\section{Appendix}%
\label{sec:appendix}

\subsection{Convergence of the quantities}%
\label{sub:Convergence of the quantities}

In this section we show that for our numerical example, $M=30$ is enough to discretize the Schrödinger operator. We take $n=2$ and $\bal = [\f 12, \f 12]$, $g=0$. We recall that $M+1$ denotes the number of planewaves that we take to numerically discretize the Hilbert space $\cH$. In Figure~\ref{fig:convergence_quantities} we display
\begin{align*}
E(M) := \f{\ab{\lambda^k(G)(M) - \lambda^k(G)(M=100)}}{\ab{\lambda^k(G)(M=100)}} 
\end{align*}
against $M$ for $k=1$. We see that the error on the eigenvalue for $M = 30$ corresponds to an error of $10^{-6}$ which seems reasonably small. Moreover, increasing $M$ does not significantly change any of the other plots presented in the numerical section.

\begin{figure}[h]
\includegraphics[width=10cm,trim={0cm 0cm 0cm 0cm},clip]{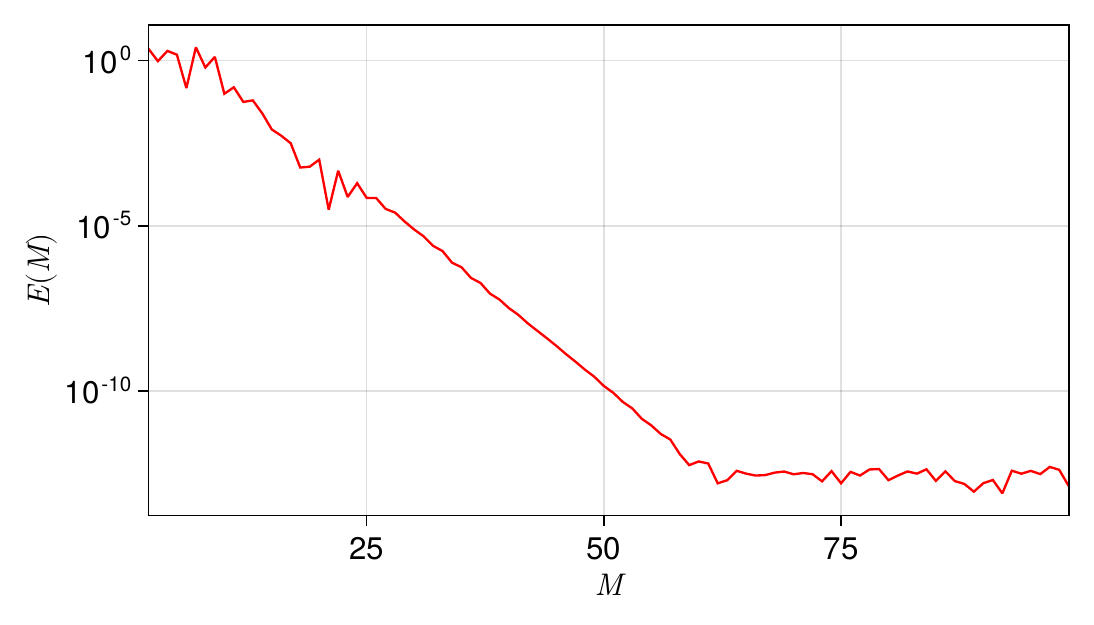}
\caption{Convergence of eigenvalues as $M$ increases.}\label{fig:convergence_quantities}
\end{figure}

\subsection{$K_{ij}$'s from $K_i$'s}%
\label{sub:Kij_from_Ki}

Here we will consider that one does not know $K_{ij}$ but only ${K_j := K(G_j)}$ using definition~\eqref{def:pseudoinv}, i.e. the quantities known in standard perturbation theory. This is not an issue since we can easily obtain approximations of $K_{ij}$ by $K_j$ if $\lambda_i$ and $\lambda_j$ are close, using the following lemma that is based on perturbative arguments. We recall that in numerical practice, the intermediate steps to compute $K_i$'s enable to compute the $K_{ij}$'s. For $j\in \{1,\dots,n\}$, we define $\lambda_j := \lambda^k(G_j)$.
\begin{lemma}[Deducing $K_{ij}$ from $K_j$]\label{lem:Kij_Kj}
If 
\begin{align*}
\ab{\lambda_i - \lambda_j} < \norm{\hzmm K_j\hzmp}^{-1},
\end{align*}
then
\begin{align}\label{eq:Kij_from_Ki}
K_{ij} = K_j \pa{1 + \pa{\lambda_i - \lambda_j}K_j}^{-1} = K_j \sum_{m=0}^{+\infty} \pa{\lambda_j - \lambda_i}^m K_j^m
\end{align}
\end{lemma}

\begin{proof}[Proof of Lemma \ref{lem:Kij_Kj}]
	We first show that
	\begin{align}\label{eq:res_diff}
	K_j - K_{ij} = \pa{\lambda_i - \lambda_j} K_j K_{ij}.
	\end{align}
	We define $P^\perp_j := 1 - P_j$, we have 
\begin{align*}
K_j = P_j^\perp \pa{\lambda_j - H(G_j)}^{-1} P_j^\perp, \qquad \qquad K_{ij} = P_j^\perp \pa{\lambda_i - H(G_j)}^{-1} P_j^\perp,
\end{align*}
	where we used an abuse of notation because $\pa{\lambda_j - H(G_j)}^{-1}$ should be defined as a pseudo-inverse, but the pseudo-inverse is equal to this quantity on $P^\perp_j \cH$. We have
\begin{align*}
	K_j - K_{ij} &= P^\perp_j \pa{\pa{\lambda_j - H(G_j)}^{-1} - \pa{\lambda_i - H(G_j)}^{-1}} P^\perp_j \\
		     &= \pa{\lambda_i - \lambda_j} P^\perp_j \pa{\lambda_j - H(G_j)}^{-1} \pa{\lambda_i - H(G_j)}^{-1} P^\perp_j \\
		     &= \pa{\lambda_i - \lambda_j} P^\perp_j \pa{\lambda_j - H(G_j)}^{-1} P^\perp_j P^\perp_j \pa{\lambda_i - H(G_j)}^{-1} P^\perp_j \\
		     &= \pa{\lambda_i - \lambda_j} K_j K_{ij},
\end{align*}
where we used that $P^\perp_j \pa{\lambda_j - H(G_j)}^{-1} P_j = P_j \pa{\lambda_j - H(G_j)}^{-1} P^\perp_j = 0$. From~\eqref{eq:res_diff}, we write $K_{ij} \pa{1 + \pa{\lambda_i - \lambda_j} K_j} = K_j$ so $K_{ij}  = K_j \pa{1 + \pa{\lambda_i - \lambda_j} K_j}^{-1}$. With a manipulation similar as in~\eqref{eq:trick}, we have
\begin{align*}
\nor{K_{ij}}{e} \le \f{\nor{K_j}{e}}{1 - \ab{\lambda_i - \lambda_j} \norm{\hzmm K_j \hzmp}},    
\end{align*}
and we can deduce~\eqref{eq:Kij_from_Ki}.
\end{proof}

\bibliographystyle{siam}
\bibliography{multipoint_perturbation}

\end{document}